\documentclass[a4paper]{article}

\usepackage{amssymb,amsthm,amsmath}

\usepackage{graphicx}
\usepackage{color}

\usepackage{authblk}

\usepackage{enumitem}

\newtheorem{theorem}{Theorem}[section]
\newtheorem{lemma}[theorem]{Lemma}
\newtheorem{corollary}[theorem]{Corollary}
\newtheorem{prop}[theorem]{Proposition}

\theoremstyle{definition}
\newtheorem{definition}[theorem]{Definition}

\theoremstyle{remark}
\newtheorem{remark}[theorem]{Remark}

\numberwithin{equation}{section}

\newcommand{\calA}{\mathcal{A}}

\newcommand{\calH}{\mathcal{H}}
\newcommand{\calR}{\mathcal{R}}
\newcommand{\alg}[1]{\mathcal{#1}}

\newcommand{\diam}[0]{\operatorname{diam}}

\newcommand{\caA}{{\mathcal A}}
\newcommand{\caB}{{\mathcal B}}
\newcommand{\caH}{{\mathcal H}}
\newcommand{\caK}{{\mathcal K}}
\newcommand{\caP}{{\mathcal P}}
\newcommand{\caU}{{\mathcal U}}
\newcommand{\bbN}{{\mathbb N}}
\newcommand{\bbR}{{\mathbb R}}
\newcommand{\nan}{\mathbb{N}}
\newcommand{\lv}{\left \vert}
\newcommand{\rv}{\right \vert}
\newcommand{\lV}{\left \Vert}
\newcommand{\rV}{\right \Vert}

\newcommand{\lmk}{\left (}
\newcommand{\rmk}{\right )}
\newcommand{\pg}{\caP_{0}(\Gamma)}
\newcommand{\unit}{\mathbb I}

\newcommand{\id}{\mathop{\mathrm{id}}\nolimits}
\newcommand{\Ad}{\mathop{\mathrm{Ad}}\nolimits}

\begin{document}
\title{The split and approximate split property in 2D systems: stability and absence of superselection sectors}

\author[1]{Pieter Naaijkens}
\affil[1]{School of Mathematics, Cardiff University, United Kingdom}
\author[2]{Yoshiko Ogata}
\affil[2]{Graduate School of Mathematical Sciences, The University of Tokyo, Japan}

\date{\today}

\maketitle

\begin{abstract}
The split property of a pure state for a certain cut of a quantum spin system can be understood 
as the entanglement between the two subsystems being weak.
From this point of view, we may say that if it is not possible to transform a state $\omega$
via sufficiently local automorphisms (in a sense that we will make precise)
into a state satisfying the split property,
then the state $\omega$ has a long-range entanglement.
It is well known that in 1D, gapped ground states have the split property with respect to cutting the system 
into left and right half-chains.
In 2D, however, the split property fails to hold for interesting models
such as Kitaev's toric code. In fact, we will show that this failure is the reason that 
anyons can exist in that model.
 
There is a folklore saying that the existence of
anyons, like in the toric code model, implies long-range entanglement of the state.
In this paper, we prove this folklore in an infinite dimensional setting.
More precisely, we show that long-range entanglement, in a way that we will define precisely, is a necessary condition to have non-trivial superselection sectors.
Anyons in particular give rise to such non-trivial sectors.
States with the split property for cones, on the other hand, do not admit non-trivial sectors.

A key technical ingredient of our proof is that under suitable assumptions on locality, the automorphisms generated by local interactions can be ``approximately factorized.''
    That is, they can be written as the tensor product of automorphisms localized in a cone and its complement respectively, followed by an automorphism acting near the ``boundary'' of $\Lambda$, and conjugation with a unitary.
This result may be of independent interest.
This technique also allows us to prove that the \emph{approximate} split property, a weaker version of the split property that \emph{is} satisfied in e.g. the toric code, is stable under applying such automorphisms.
\end{abstract}

\section{Introduction}
A pair $(N,M)$ of commuting von Neumann algebras is called \emph{split} if there is a Type I factor $F$ such that $N \subset F \subset M'$~\cite{DoplicherL}.
In applications to physics typically $N$ and $M$ are generated by local observables located in two disjoint (or, in relativistic theories, spacelike separated) regions $\Lambda_1$ and $\Lambda_2$.
The split property then can be interpreted as a type of statistical independence between regions.
More precisely, one can locally prepare a normal state $\varphi$ such that restricted to measurements in $\Lambda_i$ we have $\varphi(AB) = \varphi_1(A)\varphi_2(B)$, for given normal states $\varphi_i$ on the algebra generated by observables localized in $\Lambda_i$~\cite{Werner}.
In particular, it means that there is no entanglement between the two parts.
Alternatively, the Type I factor allows us to find a tensor product decomposition of the Hilbert space, with the algebras $N$ and $M$ acting on distinct factors.
Such a decomposition is far from obvious in systems with infinitely many degrees of freedom and may even not exist for a given bipartition of the system.
Early applications have been in algebraic quantum field theory~\cite{BuchholzW}, for example in the study of entanglement properties of the vacuum~\cite{SummersW}.

More recently the split property has found applications in the classification of phases of 1D gapped quantum spin systems.
Under quite general conditions one can show that the split property holds in ground state representations.
In particular, Matsui~\cite{Matsui13} showed that if $\omega$ is a pure ground state of a gapped local Hamiltonian (on the chain), it satisfies the split property in the sense that $\omega$ is quasi-equivalent to $\omega_L \otimes \omega_R$. 
Here $\omega_L$ (resp. $\omega_R$) is the ground state restricted to the left (resp. right) half-chain $\caA_L$ ($\caA_R$).
In this case this is equivalent to saying that the inclusion $\pi_\omega(\caA_L)'' \subset \pi_\omega(\caA_R)'$ is split in the sense above, where $\pi_\omega$ is a GNS representation for $\omega$~\cite{Matsui01} (see also~\cite[Remark 1.5]{Ogata19b}).

The split property can then be used to define a $H^2(G, U(1))$-index for a unique gapped ground state on
a quantum spin chain with finite group on-site symmetry ~\cite{Ogata19a}, as well as
${\mathbb Z}_2$-valued index for reflection  symmetry, generalizing a construction by Pollmann \emph{et al.}~\cite{PollmannTBO} for matrix product states.
The index was used to prove a general Lieb-Schultz-Mattis type theorem in \cite{OgataTT}.

For fermionic chains, the split property for a unique gapped ground state  is proven in \cite{Matsui20}.
Bourne and Schulz-Baldes and independently  Matsui
introduced a $\mathbb{Z}_2$-index for \emph{fermionic} chains {\it without symmetry}~\cite{BourneSB,Matsui20}.
A classification of SPT-phases with on-site symmetry in 1D fermionic chain based on the split property
was carried out in ~\cite{BO}. There, a ${\mathbb Z}_2\times H^1(G,{\mathbb Z}_2)\times H^2(G, U(1)_{\mathfrak p})$-valued
index was found using the split property.

The split property is essential in all these constructions: it allows one to factor the Hilbert space into a tensor product with the left half-chain acting on one factor, and the right half-chain on the other.
The Type I factor $F$ is such that $F \simeq \alg{B}(\mathcal{H}_L) \otimes I$ with respect to this decomposition.
This can then be used to extend a symmetry $\beta_L$ of the spin chain to an automorphism of $F$, which by Wigner's theorem can be implemented by a (anti-)unitary.
This in turn can be used to define an index.

In higher dimensions the situation is much more complicated, and the split property fails to hold in interesting models.
For example, consider Kitaev's toric code model~\cite{KitaevQD}.
Then one can consider a cone-like region (extending to infinity) $\Lambda$ and its complement, as an analogue of the two half-chains in 1D.
It turns out that the translation invariant ground state $\omega$ of the toric code is \emph{not} split with respect to this bipartition~\cite{Naaijkens12,FiedlerN}, in contrast with the 1D case discussed above.
In fact, one of the goals of the present work is to argue that the failure of the split property to hold is in fact necessary to get anyonic excitations.
More precisely, the failure of the split property to hold is because the state is long-range entangled.
Thus, our work confirms the folklore statement that long-range entanglement is a necessary condition for anyonic excitations.

It turns out that at least for abelian quantum double models a weaker version of the split property \emph{is} true.
That is, if one considers a pair of cones $\Lambda_1 \subset \Lambda_2$ whose boundaries are sufficiently far apart, there is a Type I factor $F$ such that $\pi_\omega(\alg{A}_{\Lambda_1})'' \subset F \subset \pi_\omega(\alg{A}_{\Lambda_2^c})'$~\cite{FiedlerN}.
``Sufficiently far'' depends on the model: in the abelian quantum double models, it is enough that the distance between their boundaries is greater than one.
In general, and in this paper as well, we need in addition that $\Lambda_2$ has a wider opening angle than $\Lambda_1.$
This should be compared with the setting in relativistic quantum field theory mentioned earlier, where the split property fails if the intersection of the closures of the two regions has non-empty intersection, but holds when they are spacelike separated.
This property is sometimes called the \emph{distal} or \emph{approximate} split property to distinguish it from the situation in e.g. 1D systems.
Despite being weaker than the split property, it still has important applications.
For example, in two dimensional systems the approximate split property is one of the assumptions used in relating the total quantum dimension (a property of the superselection sectors) to the index of a certain subfactor~\cite{NaaijkensKL}.
This result can be used to show one has found all superselection sectors of a given model.
A variant also plays a role in the discussion of ``approximately localized'' superselection sectors~\cite{ChaNN18}.

The interest of this paper is in these split and approximate split 
properties in 2D quantum spin systems.
Although most of our results can be straightforwardly generalised to higher dimensions, we restrict to 2D.
The reason is that we are particularly interested in applying our results to study anyons, and in higher dimensions the cone-localized sectors we consider automatically have bosonic or fermionic statistics (cf.~\cite{BuchholzF}).
We regard a state with the split property as having
small entanglement with respect to the given cut.
From this point of view, a 
state  which cannot be transformed 
into a split state via quasi-local automorphisms
has long-range entanglement.
Or to be more precise, we consider a slightly more restrictive class of automorphisms which we call \emph{quasi-factorizable}.
(See subsection \ref{qasubsec} for the definition of quasi-local automorphisms and their importance 
in  the theory of gapped ground state phases.)
Anyons, if they exist, can be identified with superselection sectors of the model (see Section~\ref{sec:select} for an introduction).
 We show that the existence of a non-trivial superselection sector of a state $\omega$ implies
 that the state $\omega$ is long-range entangled.
That is, long-range entanglement is a necessary condition to have non-trivial anyons.
Moreover, this is stable under applying ``quasi-factorizable'' automorphisms, defined below.
  For a class of Hamiltonians consisting of local commuting projectors, Haah~\cite{Haah} introduced 
 an ingenious index such that it having a non-trivial value implies 
 that one needs a quantum circuit with depth on the order of the system size to transform into product states.
 Our result is in accordance with these results.
In general, the split property itself in 2D is not stable under quasi-local automorphisms.
We show, however,
 the approximate split property is stable under it.

 The key technical ingredient for the proof is a factorization property
of quasi-local automorphisms $\alpha_s$.
This result may be of independent interest.
More precisely, we show that under mild assumptions, $\alpha_s$ is \emph{quasi-factorizable} in the following sense.
In the definition below, $\Gamma$ is the set of all sites of the system, and for any subset $\Lambda \subseteq \Gamma$, $\calA_\Lambda$ is the corresponding quasi-local $C^*$-algebra of observables localized in $\Lambda$ (see below).
\begin{definition}\label{def:quasifactor}
Let $\alpha$ be an automorphism of $\calA_\Gamma$ and consider an inclusion of cones
\[
        \Gamma_1' \subset \Lambda \subset \Gamma_2' \subset \Gamma.
\]
We say that $\alpha$ is \emph{quasi-factorizable} with respect to this inclusion if there is a unitary $u \in \calA$ and automorphisms $\alpha_\Lambda$ and $\alpha_{\Lambda^c}$ of $\calA_{\Lambda}$ and $\calA_{\Lambda^c}$ respectively, such that 
\[
    \alpha = \Ad(u) \circ \widetilde{\Xi} \circ (\alpha_{\Lambda} \otimes \alpha_{\Lambda^c}),
\]
where $\Xi$ is an automorphism on $\Gamma_2' \setminus \Gamma_1'$ and $\Lambda^c := \Gamma \setminus \Lambda$.
\end{definition}
The key advantage is that one can replace the ``exponential tails'' of $\alpha_s$ by strict locality, up to conjugation with a unitary in $\alg{A}_\Gamma$.
In for example the sector theory, such strict locality is very useful, and one is only interested in representations up to unitary equivalence.

This factorization property was first used in ~\cite{Ogata19a},
in proving the stability of the index of 1D SPT.
Following this idea, in ~\cite{Moon} the stability of split property in 1D was shown.
Its $2$-dimensional version is essential here,
but an extra complication is that in 2D or higher, the boundary between the regions we will consider is infinite.
This makes locality estimates much more subtle.
Coincidentally, this more complicated geometry is also a key reason why Matsui's result on the split property for 1D spin chains~\cite{Matsui13} does not generalize to higher dimensions.
A special case of the 2D-version (with respect to cone like regions with common apex) of the factorization property is also used in
~\cite{Ogata21}, to define a $H^{3}(G,U(1))$-valued index and to show its stability.

In Section~\ref{sec:prelim} we fix notation and recall some basic facts about Lieb-Robinson bounds and quasi-local maps, and give a brief overview of the relation between anyons and superselection sectors.
Then, in Section~\ref{sec:stable}, we prove the factorization property of
quasi-local automorphisms in a general setting.
In Section~\ref{sec:lre} we consider states in 2D which are quasi-equivalent to a product state, and hence satisfy the strict split property.
In particular, we show that the states in this gapped quantum phase have trivial superselection structure.
Finally, in Section~\ref{sec:approxsplit} we show that our main technical result applies to a natural class of quasi-local automorphisms, and use this to show that the approximate split property is stable in such models.

\emph{Acknowledgments.} PN was supported in part by funding from the European Union’s Horizon 2020 research and innovation program under the European Research Council (ERC) Consolidator Grant GAPS (No. 648913). YO is supported is supported in part by JSPS KAKENHI Grant Number 16K05171 and 19K03534.
She was also supported by JST CREST Grant Number JPMJCR19T2.

\section{Preliminaries}\label{sec:prelim}
We first fix the setting and introduce the main definitions.
A key part is played by quasi-local maps and Lieb-Robinson bounds.
For a state-of-the-art overview of the topic see~\cite{NSY};
for our purpose the most relevant facts will be recalled here.
We largely adopt the notation of~\cite{NSY}.
We assume basic familiarity with the operator algebraic formulation of quantum spin systems (see e.g.~\cite{BratteliR1,BratteliR2}).

Let $(\Gamma, d)$ be a countable metric space
which is $\nu$-regular i.e.,
\begin{align}\label{nreg}
\sup_{x\in \Gamma}\lv b_{x}(n)\rv\le \kappa n^{\nu},\quad 1\le n\in \bbN,
\end{align}
for some constant $\kappa>0$.
Here, we used the notation
\begin{align}
 b_{x}(n):=\left\{
 y\in\Gamma \mid d(x,y)\le n
 \right\}.
\end{align}
In concrete applications we typically consider $\Gamma = \mathbb{Z}^\nu$ (or its edges) with the usual metric, but for now we keep the discussion as general as possible.

Let $\caP_{0}(\Gamma)$ be
the set of all finite subsets of $\Gamma$.
For $\Lambda \in \caP_{0}(\Gamma)$ we set
\begin{align}\label{ag}
\caA_{\Lambda}:=\bigotimes_{x\in\Lambda}\caB(\caH_{x}),
\end{align}
where $\caH_{x}$ are finite dimensional Hilbert spaces whose dimensions are uniformly bounded:
\begin{align}\label{agg}
\sup_{x\in\Gamma} \dim \caH_{x}<\infty.
\end{align}
If $\Lambda_1 \subset \Lambda_2$ there is a natural inclusion of algebras, and hence we can write
\begin{align}
\caA_{\Gamma}^{\rm loc} :=\bigcup_{\Lambda\in \caP_{0}(\Gamma)}\caA_{\Lambda}
\end{align}
for the algebra of local observables.
To get the $C^*$-algebra $\caA_{\Gamma}$ of quasi-local observables we take the norm closure of $\caA_{\Gamma}$.
In general, if $\Lambda \subset \Gamma$ is any subset of $\Gamma$, $\caA_\Lambda$ is the norm closure of $\bigcup_{\Lambda_0 \subset \Lambda, \Lambda_0 \in \caP_{0}(\Gamma)} \caA_{\Lambda_0} $.
We denote by $ \caU\lmk \caA_{\Gamma}\rmk$ the set of all unitaries in $ \caA_{\Gamma}$.

For any subset $X$ of $\Gamma$, we denote by $\Pi_{X}$ the conditional expectation
onto $\caA_{X}$ given by the tracial state on $\caA_{X^{c}}$.
These maps will be used to approximate quasi-local observables by local ones.
For any $m\in\nan\cup \{0\}$ and $X\subset\Gamma$, we set
\begin{align}
X(m):=\left\{ x\in\Gamma\mid d(x,X)\le m\right\}.
\end{align}
Furthermore, we define
\begin{align}
\Delta_{X(m)}:=\Pi_{X(m)}-\Pi_{X(m-1)},\quad m\in\nan,\; X\subset \Gamma.
\end{align}
Note that we have
\begin{align}
\lV
\Delta_{X(m)}\lmk A\rmk
\rV
\le 2 \lV A\rV,\quad A\in\caA_{\Gamma},
\end{align}
since $\Pi$ is a projection.

\subsection{Split property}
We will be interested in the split property with respect to different regions of $\Gamma$, leading to the following definition.
\begin{definition}
\label{defn:split}
Let $\Gamma_1 \subset \Gamma_2 \subset \Gamma$ and $\omega$ a pure state of $\caA_\Gamma$.
Then we say that $\omega$ is \emph{split} with respect to the inclusion $\Gamma_1 \subset \Gamma_2$ if there is a Type I factor $F$ such that
\begin{equation}
	\pi( \caA_{\Gamma_1} )'' \subset F \subset \pi( \caA_{\Gamma_2})'', 
\end{equation}
where $\pi$ is a GNS representation for $\omega$.
\end{definition}

Conjugating with a unitary does not affect the split property.
Furthermore, one would expect that automorphisms of $\caA_{\Gamma_1}$ and $\caA_{\Gamma_2^c}$ have no effect on the existence of the Type I factor $F$.
We can even allow for a non-trivial automorphism on a ``widening'' of the region $\Gamma_2 \setminus \Gamma_1$, at the expense of shrinking (resp. growing) the two regions in the definition of the split property.
This is the idea behind the next proposition.
\begin{prop}\label{prop:splitstable}
Let $\Gamma_0\subset\Gamma_1\subset\Gamma_2\subset\Gamma_3$
be a sequence of subsets in $\Gamma$.
Let $\omega$ be a pure state on
$\caA_{\Gamma}$ and suppose that it is split with respect to $\Gamma_1 \subset \Gamma_2$.
Let $\alpha$ be an automorphism of $\caA_{\Gamma}$.
Let  $\alpha_{\Gamma_2^c}$,
$\alpha_{\Gamma_2\setminus \Gamma_1}$,
$\alpha_{\Gamma_1}$ be automorphisms
of $\caA_{\Gamma_2^c}$,
$\caA_{\Gamma_2\setminus \Gamma_1}$,
$\caA_{\Gamma_1}$
respectively.
Define an automorphism $\tilde\alpha$ of $\caA_{\Gamma}$ by
\begin{align}
\tilde \alpha:=\alpha_{\Gamma_2^c}\otimes \alpha_{\Gamma_2\setminus \Gamma_1}\otimes
\alpha_{\Gamma_1}.
\end{align}
Suppose moreover that there is an automorphism
$\beta_{\Gamma_3\setminus \Gamma_0}$ of $\caA_{\Gamma_3\setminus \Gamma_0}$
and a unitary $u\in\caA_{\Gamma}$ such that
\begin{align}\label{eq:factorize}
\alpha=\Ad(u)\circ
\tilde\alpha\circ\lmk\tilde  \beta_{\Gamma_3\setminus \Gamma_0}\rmk,
\end{align}
where $\tilde  \beta_{\Gamma_3\setminus \Gamma_0} =\beta_{\Gamma_3\setminus \Gamma_0}\otimes{\id_{(\Gamma_3\setminus \Gamma_0)^c}}$.
Then $\omega \circ \alpha$ is split for the inclusion $\Gamma_0 \subset \Gamma_3$.
In fact, if $(\caH,\pi,\Omega)$ is a GNS triple of $\omega$ and $F$ the interpolating factor from Def.~\ref{defn:split},
we can choose $\tilde F=\Ad(\pi(u))(F)$ as the interpolating Type I factor in the GNS representation $\pi \circ \alpha$ for the state $\omega \circ \alpha$.
\end{prop}
\begin{proof}
We have
\begin{align}
\pi\circ\tilde \alpha\lmk\caA_{\Gamma_0}\rmk''
=\pi\circ\alpha_{\Gamma_1}\lmk\caA_{\Gamma_0}\rmk''
\subset \pi\lmk\caA_{\Gamma_1}\rmk''\subset F.
\end{align}
We also have
$
 \tilde\alpha^{-1} (\caA_{\Gamma_2})
= 
 \caA_{\Gamma_2}
\subset \caA_{\Gamma_3}$,
and hence 
$\pi\lmk \caA_{\Gamma_2}\rmk
\subset \pi \circ \tilde\alpha\lmk \caA_{\Gamma_3}\rmk$.
Therefore we have
\begin{align}
\pi\circ\tilde \alpha(\caA_{\Gamma_3})'
\subset \pi(\caA_{\Gamma_2})'
\subset 
F',
\end{align}
and by taking commutants
$F\subset
\pi\circ\tilde \alpha(\caA_{\Gamma_3})''$.
Hence we obtain
\begin{align}
\pi\circ\tilde \alpha\lmk\caA_{\Gamma_0}\rmk''\subset 
F\subset
\pi\circ\tilde \alpha(\caA_{\Gamma_3})''.
\end{align}
Note that by assumption and the fact that $\tilde{\beta}_{\Gamma_3\setminus \Gamma_0}$ acts trivially on $\caA_{\Gamma_0}$,
\begin{align}
\alpha\lmk\caA_{\Gamma_0}\rmk
=\Ad(u)\circ
\tilde\alpha\circ\tilde\beta_{\Gamma_3\setminus \Gamma_0}\lmk
\caA_{\Gamma_0}\rmk
=\Ad(u)\circ
\tilde\alpha\circ\lmk
\caA_{\Gamma_0}\rmk,
\end{align}
and similar with $\caA_{\Gamma_0}$ replaced by $\caA_{\Gamma_3}$.
Hence we have
\begin{equation}
\begin{split}
\lmk \pi\circ\alpha\lmk\caA_{\Gamma_0}\rmk\rmk''
= & \Ad\lmk\pi(u)\rmk\lmk \pi\circ \tilde\alpha\lmk\caA_{\Gamma_0}\rmk''\rmk 
\subset \Ad\lmk\pi(u)\rmk\lmk F\rmk \\ 
&\subset \Ad\lmk\pi(u)\rmk\lmk  \lmk \pi\circ\tilde \alpha\lmk\caA_{\Gamma_3}\rmk\rmk''\rmk
=\lmk\pi\circ\alpha\lmk\caA_{\Gamma_3}\rmk\rmk''.
\end{split}
\end{equation}
This completes the proof.
\end{proof}

Note that the condition on $\alpha$ implies that $\alpha^{-1}$ is quasi-factorizable for the inclusion $\Gamma_0 \subset \Gamma_2 \subset \Gamma_3$ in the sense of Definition~\ref{def:quasifactor}.
The main technical contribution of the paper consists in proving that the quasi-local automorphisms $\alpha_s$ admit a decomposition as in~\eqref{eq:factorize} of the proposition.

\subsection{Sector theory}\label{sec:select}
The present work is at least partly motivated by superselection sector theory, in the sense of Doplicher, Haag and Roberts (DHR). See~\cite{Araki99,HaagLQP} for an introduction.
In two dimensional systems with long-range topological order, there is the possibility of quasi-particles with braided exchange statistics.
Typical examples of such models are Kitaev's quantum double models~\cite{KitaevQD} and the Levin-Wen string-net models~\cite{LevinW}.
Mathematically, the algebraic properties of the anyons are described by a braided tensor category~\cite{Wang}.
Thus, the question is how one can extract this tensor category from first principles.

Typical methods to extract the braided tensor category from a ground state rely quite heavily on certain properties (e.g. symmetries) of the underlying model, and are therefore less suitable for a general analysis.
In fact, in finite systems it is not always clear how to even define a single anyonic excitation, in particular once one loses strict locality as a result of perturbations.
We therefore take a different approach, motivated by DHR sector theory in algebraic quantum field theory~\cite{HaagLQP}, in which one in principle can recover the full anyon structure from a few general and physically motivated principles.
The idea of a superselection sector stems from the observation that it appears to be impossible to make coherent superpositions between certain states, in particular when they carry a different `charge' or `anyon type'.
Mathematically this phenomenon is related to the existence of non-equivalent representations of the algebra of observables.
One way to interpret this is to think of charge conservation: with local operations it is not possible to change the total charge of the system.
In particular, say we create a conjugate pair of anyons (thus preserving the total charge) from the ground state, and move one far away.
Then acting \emph{locally} the total charge in that region cannot be changed.
Or, to give an example, it is impossible to create a vector state describing a single charged anyon in the ground state representation of a topologically ordered model, using quasi-local observables only.
Equivalently, it is not possible to create coherent superpositions of disjoint states (cf.~\cite[Thm. 6.1]{Araki99}).

The $C^*$-algebra $\caA_\Gamma$ has many inequivalent representations, but most of them are not physically relevant.
Hence we need a selection criterion to select the relevant representations that correspond to charged states (that is, states describing single anyon excitations).
It is perhaps helpful to illustrate how this works in the prototypical example of the toric code~\cite{KitaevQD}.
We refer to~\cite{Naaijkens11,FiedlerN} for details on the following discussion.
In the thermodynamic limit, one can show that there is a translation invariant ground state, uniquely characterised by the condition that $\omega_0(A_s) = \omega_0(B_p) = 1$.
Here $A_s$ and $B_p$ are the `star' and `plaquette' operators appearing in the Hamiltonian for the toric code.
It is well-known that one can define `string operators' $F_\xi$ that create a pair of excitations (anyons) when acting on the ground state of the toric code.
Note that the excitations at the end of the path $\xi$ are conjugate to each other, so that the total charge of the anyons created by this operator is trivial.
Thus $A \mapsto \omega_0(F_\xi A F_\xi^*)$ is a state describing a pair of anyons.
To get a state describing a \emph{single} anyon, one can take the limit where one end of the path is sent off to infinity.
This converges, and one can show that the resulting state is inequivalent to $\omega_0$.
Moreover, by construction, this state can be interpreted as describing a single anyon, located at the endpoint that was kept fixed.

The corresponding GNS representation $\pi$ has additional properties, reminiscent of the topological charges in algebraic quantum field theory~\cite{BuchholzF}.
For example, suppose that the paths $\xi$ in the construction above all lie in some cone $\Lambda$.
Then it is easy to show that \emph{outside} of the cone the GNS representation for $\omega$ is unitarily equivalent to the ground state representation $\pi_0$.
This means that the anyon is localized in the cone $\Lambda$.
What is less obvious is that if we choose a path going off to infinity in a different direction, the corresponding GNS representation is unitarily equivalent to $\pi$.
The same is true if we choose a different endpoint for the path $\xi$.
This property ultimately boils down to the property of the toric code that the state $\omega_0(F_\xi A F_\xi^*)$ only depends on the endpoints of the path $\xi$, and not on the path itself.
To summarise, the single anyon representation $\pi$ is irreducible, and satisfies
\begin{equation}
	\label{eq:sselect}
        \pi_0|\calA_{\Lambda^c} \cong \pi|\calA_{\Lambda^c},
\end{equation}
for \emph{any} cone $\Lambda$.\footnote{The choice of cones as localization region is merely a convenient one, motivated by space-like cones in algebraic QFT~\cite{BuchholzF}.
        What is more important is that it extends to infinity.
        This allows us to send one of the ends of a ``string operator'' creating a pair of anyons in e.g. the toric code to infinity.
        For technical reasons, we need the region to ``widen'' towards infinity, so that any finite region can be transported into it, and that the region has no holes.
        An advantage of cones is that they are easy to parametrise, cf.~\cite{Ogata21a}.
}
Here $\pi_0$ is the (reference) ground state representation, and $\cong$ denotes unitary equivalence of the representation restricted to $\calA_{\Lambda^c}$, the observables localized outside of the cone $\Lambda$.
Since the criterion is required to hold for \emph{any} cone, the localization region can be moved around.
This is called \emph{transportability} of the charges, and we say that the charge is \emph{transportable}
(see e.g.~\cite[Section IV.2]{HaagLQP}).
For the toric code, it is straightforward to construct four different inequivalent representations that satisfy this property, corresponding to the four anyon types of the model.

For general topologically ordered models, one expects the charges to have the same localization properties (for example based on the string operators that are typical for such models).
Thus, in general, a reasonable approach is to take a ground state representation $\pi_0$, and identify irreducible representations $\pi$ satisfying~\eqref{eq:sselect} with the charges (or, anyons) of the theory.
A sector is then a (unitary) equivalence class of representations $\pi$ satisfying the selection criterion.
The \emph{trivial} sector is the equivalence class containing the reference representation $\pi_0$.
Later we will slightly relax the criterion~\eqref{eq:sselect} to require only quasi-equivalence.

It is perhaps surprising that by just imposing this single selection criterion, we obtain a very rich structure.
In fact, based on the DHR program and using a technical property called Haag duality, one can show that the set of representations satisfying this criterion has the structure of a braided tensor category~\cite{BuchholzF,Naaijkens11,Ogata21a}.
In addition, in concrete models such as the toric code there are natural candidates to construct representations $\pi$ satisfying the criterion, even without resorting to Haag duality, as outlined above.
Moreover, one can prove that these representations are the only ones satisfying the selection criterion~\eqref{eq:sselect}, and it follows that the category is equivalent to the representation of the quantum double of the group $G = \mathbb{Z}_2$, as expected~\cite{NaaijkensKL}.
This result can be generalised to abelian quantum double models~\cite{FiedlerN}.
Thus, we take the viewpoint that each type of anyon gives rise to an equivalence class of representations $\pi$ satisfying~\eqref{eq:sselect}.

The split property enters the analysis in various ways.
We first note that the topological phenomena in our systems of interest, in particular the existence of anyons, are believed to be due to the presence of \emph{long-range entanglement}~\cite{ChenGW}.
Product states exhibit no entanglement, and hence should be in the trivial phase without any anyons.
A state with long-range entanglement is then roughly speaking a state that cannot be transformed into a product state by applying a finite sequence of local unitaries throughout the system.
Consider the case where we have a pure state $\omega = \omega_\Lambda \otimes \omega_{\Lambda^c}$ that is a product state with respect to a cone $\Lambda$ and its complement.
It is easy to see (see Section~\ref{sec:lre}) that in this case $\pi_\omega(\caA_\Lambda)''$ is a Type I factor and the inclusion $\pi_\omega(\caA_\Lambda)'' \subset \pi_\omega(\caA_{\Lambda^c})'$ therefore is split.
In Section~\ref{sec:lre} we show that in this case the sector theory is trivial:
any representation $\pi$ satisfying~\eqref{eq:sselect} is a direct sum of copies of the reference representation $\pi_0$.
That is, we only have the trivial charge or anyon.
This corroborates the notion that the sector theory is a good invariant for topological phases by proving that indeed states without long-range entanglement have a trivial sector structure.
Indeed, we will prove that this still is the case for pure states $\omega$ such that $\omega \circ \alpha$ is quasi-equivalent to a product state.
Here, $\alpha$ is a quasi-factorizable automorphism, which can be seen as a generalization of finite-depth quantum circuits to infinite systems.
This result also explains why in models such as the toric code, which \emph{do} have a non-trivial sector theory, we only have a weaker form of the split property, where we have to consider an inclusion $\Lambda_1 \subset \Lambda_2$ of cones whose boundaries are sufficiently far apart~\cite{Naaijkens12}.

This weaker form of the split property also plays a role in the analysis in~\cite{NaaijkensKL}, where the index of a certain subfactor is shown to be related to the total quantum dimension of the sectors.
This result can be used to show that a given list of sectors is complete.
It also is necessary in showing that approximately localized sectors, a generalisation of the notion of a sector discussed above, is stable under applying a path of quasi-local automorphisms~\cite{ChaNN18}.
In either case, the split property for an inclusion $\Lambda_1 \subset \Lambda_2$ allows us to obtain a tensor product decomposition of the ground state Hilbert space such that observables in $\caA_{\Lambda_1}$ and those in $\caA_{\Lambda_2^c}$ act on the distinct factors.
In contrast to finite systems such a decomposition need not exist if the split property fails to hold.
This decomposition can then be used to approximately localize endomorphisms or observables~\cite{ChaNN18}.
This plays a crucial role in the proof of stability of superselection sectors.
Although the proof only requires a variant of the split property to hold at one point along the path of gapped Hamiltonians, it is nevertheless important to understand the stability of the split property itself.

\subsection{Quasi-local maps}\label{qasubsec}
In the classification problem of gapped ground state phases,
we say that two states are in the same phase if they can be realized as ground states of gapped Hamiltonians that can be connected via a continuous (or, for technical reasons, $C^1$) path, in such a way that the energy gap does not close along the path.
Using the spectral flow~\cite{BachmannMNS}, an adaptation of Hastings and Wen's quasi-adiabatic continuation~\cite{HastingsW} to the thermodynamic limit, one obtains a path of automorphisms $s \mapsto \alpha_s $ relating the ground states along the path of gapped Hamiltonians.
Its infinite system version, where a uniform gap for the local Hamiltonians can be replaced by
the spectral gap of the bulk Hamiltonian in the GNS representation was shown in \cite{MoonO19}.
Quasi-local automorphisms are essential transformation
in the theory of gapped ground state phases.

A \emph{quasi-local map} on $\caA_\Gamma$ is a map that maps strictly localized observables to observables that can still be approximately localized in a slightly larger region, with error bounds satisfying a Lieb-Robinson type of estimate.
Our discussion draws heavily on~\cite{NSY}, which in turn incorporates decades of advancements in Lieb-Robinson bounds.

Typically the quasi-local maps are obtained as the dynamics generated by some sufficiently local interaction.
The notion of ``sufficiently local'' is made precise in the following definitions.

\begin{definition}
An $F$-function $F$ on $(\Gamma, d)$
is a non-increasing function $F:[0,\infty)\to (0,\infty)$
such that
\begin{description}
\item[(i)]
$\lV F\rV:=\sup_{x\in\Gamma}\lmk \sum_{y\in\Gamma}F\lmk d(x,y)\rmk\rmk<\infty$,
and
\item[(ii)]
$C_{F}:=\sup_{x,y\in\Gamma}\lmk \sum_{z\in\Gamma}
\frac{F\lmk d(x,z)\rmk F\lmk d(z,y)\rmk}{F\lmk d(x,y)\rmk}\rmk<\infty$.
\end{description}
These are called \emph{uniform integrability} and the \emph{convolution identity}, respectively.
\end{definition}

For an $F$-function $F$ on $(\Gamma, d)$,
define a function $G_{F}$ on $t\ge 0$ by
\begin{align}\label{gfdef}
G_{F}(t):= \sup_{x\in\Gamma}\lmk
\sum_{y\in\Gamma, d(x,y)\ge t} F\lmk d(x,y)\rmk
\rmk,\quad t\ge 0.
\end{align}
Note that by uniform integrability the supremum is finite for all $t$.

Our goal is to interpolate continuously between two local interactions.
Hence we will mainly be considering \emph{paths} of local interactions, in the following sense:
\begin{definition}
A norm-continuous interaction on $\caA_{\Gamma}$ defined on an interval $[0,1]$
is a map
$\Phi:\caP_{0}(\Gamma)\times [0,1]\to \caA_{\Gamma}^{\rm loc}$ such that
\begin{description}
\item[(i)]
for any $t\in[0,1]$, $\Phi(\cdot; t):\caP_{0}(\Gamma)\to \caA_{\Gamma}^{\rm loc}$
is an interaction, and
\item[(ii)]
for any $Z\in\caP_{0}(\Gamma)$, the map
$\Phi(Z;\cdot ):[0,1]\to \caA_{Z}$
is norm-continuous.
\end{description}
\end{definition}

To ensure that the interactions induce quasi-local automorphisms we need to impose sufficient decay properties on the interaction strength.

\begin{definition}
Let $F$ be an $F$-function on $(\Gamma,d)$.
We denote by $\caB_{F}([0,1])$ the set of all
norm continuous interactions on $\caA_{\Gamma}$ defined on an interval $[0,1]$
such that the function $\lV
\Phi
\rV: [0,1]\to \bbR$ defined by
\begin{align}
\lV
\Phi
\rV(t):=
\sup_{x,y\in\Gamma}\frac{1}{F\lmk d(x,y)\rmk}\sum_{Z\in\caP_{0}(\Gamma), Z\ni x,y}
\lV\Phi(Z;t)\rV,\quad t\in[0,1],
\end{align}
is uniformly bounded, i.e., 
$\sup_{t\in[0,1]}\lV \Phi \rV(t)<\infty$.
It follows that $t \mapsto \lV \Phi \rV(t)$ is integrable, and we set 
\begin{align}
I(\Phi):=I_{1,0}(\Phi):= C_{F} \int_{0}^{1} dt\lV \Phi \rV(t) .
\end{align}
\end{definition}
We will need some more notation.
For $\Phi\in \caB_{F}([0,1])$ and $0\le m\in\bbR$, 
we introduce a path of interactions $\Phi_{m}$ by
\begin{align}\label{pm}
\Phi_{m}\lmk X;t\rmk:=|X|^{m}\Phi\lmk X;t\rmk,\quad X\in\pg,\quad t\in[0,1].
\end{align}
Next we recall that an interaction gives rise to local (and here, time-dependent) Hamiltonians, via
\begin{align}
H_{\Lambda,\Phi}(t):=\sum_{Z\subset\Lambda}\Phi(Z;t),\quad t\in[0,1].
\end{align}
We denote by $U_{\Lambda,\Phi}(t;s)$, the solution of 
\begin{align}
\frac{d}{dt} U_{\Lambda,\Phi}(t;s)=-iH_{\Lambda,\Phi}(t) U_{\Lambda,\Phi}(t;s),\quad t\in[0,1]\\
U_{\Lambda,\Phi}(s;s)=\unit.
\end{align}
We define corresponding automorphisms $\tau_{t,s}^{(\Lambda),\Phi}, \hat{\tau}_{t,s}^{(\Lambda), \Phi}$ on $\caA_{\Gamma}$ by
\begin{align}
\tau_{t,s}^{(\Lambda), \Phi}(A):=U_{\Lambda,\Phi}(t;s)^{*}AU_{\Lambda,\Phi}(t;s),\\
\hat{\tau}_{t,s}^{(\Lambda), \Phi}(A):=U_{\Lambda,\Phi}(t;s)AU_{\Lambda,\Phi}(t;s)^{*},
\end{align}
with $A \in \caA_\Gamma$. 
Note that
\begin{align}\label{inv}
\hat{\tau}_{t,s}^{(\Lambda), \Phi}={\tau}_{s,t}^{(\Lambda), \Phi},
\end{align}
by the uniqueness of the solution of the differential equation.
Using standard techniques one can prove locality estimates for time-evolved local observables in the form of Lieb-Robinson bounds, which in turn can be used to show that the local dynamics $\tau_{t,s}^{(\Lambda), \Phi}$ induce global dynamics.
Since we will make use of these facts repeatedly we recall the main points here.
\begin{theorem}[\cite{NSY}]\label{tni}
Let $(\Gamma, d)$ be a countable metric space, and let $F$ be a $F$-function on $(\Gamma, d)$.
Suppose that $\Phi\in\caB_F([0,1])$.
The following holds:
\begin{enumerate}[label={\rm\textbf{(\roman*)}}]
\item \label{it:tdlimit} The limit
\begin{align}
\tau_{t,s}^{\Phi}(A):=\lim_{\Lambda \nearrow\Gamma}\tau_{t,s}^{(\Lambda), \Phi}(A),\quad
A\in\caA_{\Gamma}, \quad t,s\in[0,1]
\end{align}
exists and defines a strongly continuous family of automorphisms on $\caA_{\Gamma}$
such that
\begin{align}
\tau_{t,s}^{\Phi}\circ\tau_{s,u}^{\Phi}=\tau_{t,u}^{\Phi},\quad \tau_{t,t}^{\Phi}=\id_{\caA_{\Gamma}}, \quad t,s,u\in[0,1].
\end{align}
\item \label{it:lr}For any $X,Y\in \caP_{0}(\Gamma)$ with $X\cap Y=\emptyset$, and $A\in \caA_{X}$, $B\in\caA_{Y}$
we have
\begin{align}
\lV
\left[
\tau_{t,s}^{\Phi}(A), B
\right]
\rV
\le \frac{2\lV A\rV\lV B\rV}{C_{F}}\lmk e^{2I(\Phi)}-1\rmk\lv X\rv G_{F}\lmk d(X,Y)\rmk.
\end{align}
If $\Lambda \in \caP_{0}(\Gamma)$ and $X \cup Y \subset \Lambda$, a similar bound holds for $\tau_{t,s}^{(\Lambda),\Phi}$.

\item \label{it:approx}
For any $X\in \caP_{0}(\Gamma)$ 
we have
\begin{align}
&\lV
\Delta_{X(m)}\lmk
\tau_{t,s}^{\Phi}(A)\rmk
\rV
\le \frac{4\lV A\rV}{C_{F}}\lmk e^{2I(\Phi)}-1\rmk\lv X\rv G_{F}\lmk m\rmk,
\end{align}
for all $\Lambda\in\caP_{0}(\Gamma)$
and $A\in \caA_{X}$.
A similar bound holds for $\tau_{t,s}^{(\Lambda),\Phi}$.

\item \label{it:localdyn}
For any $X,\Lambda\in \pg$ with $X\subset\Lambda$, and $A \in \caA_X$
we have
\begin{align}
\lV
\tau_{t,s}^{(\Lambda), \Phi}(A)-\tau_{t,s}^{\Phi}(A)
\rV
\le\frac{2}{C_{F}} \lV A\rV e^{2I(\Phi)}I(\Phi) \lv X\rv G_{F}\lmk d\lmk X,\Gamma\setminus\Lambda\rmk
\rmk.
\end{align}
\end{enumerate}
\end{theorem}
\begin{proof}
	Item~\ref{it:tdlimit} is Theorem~3.5 of~\cite{NSY}, while~\ref{it:lr} and~\ref{it:localdyn} follow from Corollary~3.6 of the same paper by a straightforward bounding of $D(X,Y)$ and the summation in eq.~(3.80) of~\cite{NSY} respectively.
	Finally,~\ref{it:approx} can be obtained using~\ref{it:lr} and~\cite[Cor. 4.4]{NSY} (see also the proof of Lemma~5.1 in the same paper).
\end{proof}

Consider the same notation and assumptions as in Theorem~\ref{tni}.
To continue we need to make additional assumptions on the function $F$.
In particular, we assume that there is an $\alpha\in(0,1)$ such that
\begin{align}\label{as:galp}
	\sum_{n=0}^{\infty} (1+n)^{2\nu+1}G_F(n)^{\alpha}<\infty,
\end{align}
where $G_F$ is as defined in~\eqref{gfdef}.
Furthermore, we assume that there is an $F$-function $\tilde F$ on $(\Gamma, d)$ such that
\begin{align}\label{as:gf}
\max\left\{
F\lmk \frac r 3\rmk, \sum_{n=[\frac r 3]}^{\infty} (1+n)^{2\nu+1}G_F(n)^{\alpha}
\right\}\le \tilde F(r).
\end{align}
With these additional assumptions we can distill the following result.
It gives us a way to apply a quasi-local automorphism to a given dynamics.
The result will generally \emph{not} be an interaction, since the interaction terms will not localized in finite regions any more.
Nevertheless, the theorem shows that we can define a proper interaction that gives the correct local Hamiltonians.
\begin{theorem}\label{tsan}
Let $(\Gamma, d)$ be a countable $\nu$-regular metric space and $F$ be an $F$-function on $(\Gamma, d)$ such that there are $\alpha$ and $\tilde{F}$ satisfying~\eqref{as:galp} and~\eqref{as:gf}.
Let $\Phi\in\caB_{F}([0,1])$ be a path of interactions such that
$\Phi_{1}\in \caB_{F}([0,1])$, where $\Phi_1$ is defined in~\eqref{pm}.
Finally, choose an increasing sequence $\Lambda_n \in \caP_0(\Gamma)$ such that $\Lambda_n \nearrow \Gamma$, and let $\tau_{t,s}^\Phi$ and $\tau_{t,s}^{(\Lambda_n),\Phi}$ be as in Theorem~\ref{tni}.

Then, with $s \in [0,1]$,
the right hand side of the following sum
\begin{align}\label{eq:psis}
\Psi^{(s)}\lmk
Z, t
\rmk:=
\sum_{m\ge 0} \sum_{X\subset Z,\; X(m)=Z}
\Delta_{X(m)}\lmk
\tau_{t,s}^\Phi\lmk
 \Phi\lmk X; t\rmk
\rmk
\rmk
\end{align}
defines an interaction $\Psi^{(s)}\in\caB_{\tilde F}([0,1])$.
 Furthermore, the formula 
 \begin{align}\label{eq:psisn}
	 \Psi^{(n)(s)}\lmk
Z, t
\rmk:=
\sum_{m\ge 0} \sum_{X\subset Z, X(m)\cap\Lambda_{n}=Z}
\Delta_{X(m)}\lmk
\tau_{t,s}^{(\Lambda_n, \Phi)}\lmk
 \Phi\lmk X; t\rmk
\rmk
\rmk
\end{align}
defines $\Psi^{{(n)(s)}}\in \caB_{\tilde F}([0,1])$ such that
$\Psi^{(n)}\lmk
Z, t
\rmk=0$ unless $Z\subset \Lambda_{n}$,
and satisfies
\begin{align}\label{psio}
\tau_{t,s}^{(\Lambda_n), \Phi} \lmk H_{\Lambda_n, \Phi}(t)\rmk
=H_{\Lambda_n, \Psi^{(n)}}(t).
\end{align}
For any $t,u\in[0,1]$, we have
\begin{align}\label{convconv}
\lim_{n\to\infty}\lV
\tau_{t,u}^{\Psi^{(n)(s)}}\lmk A\rmk
-\tau_{t,u}^{\Psi^{(s)}}\lmk A\rmk
\rV=0,\quad A\in\caA_{\Gamma}.
\end{align}
\end{theorem}
\begin{proof}
If $Z$ is a finite set, we see that the right-hand side of~\eqref{eq:psis} contains only finitely many terms and hence is well-defined.
Moreover, because of the $\Delta_{X(m)}$, it follows that $\Psi^{(s)}(Z,t) \in \caA_Z$.
Since $\tau_{t,s}$ is in automorphism we see that $\Psi^{(s)}(Z,t)$ is self-adjoint, and hence defines an interaction.
That this interaction is in $\caB_{\tilde{F}}([0,1])$ follows then from Theorem 5.17(i) of~\cite{NSY}.
The conditions of this theorem can be verified using Theorem~\ref{tni}, where in the notation of~\cite{NSY} we have $p=0$ and $q=r=1$.

Similarly, equation~\eqref{eq:psisn} defines an interaction, and~\eqref{psio} can be verified by an explicit calculation, if we note that $\tau_{t,s}^{(\Lambda_n), \Phi}(\Phi(X;t))$ is in $\caA_{\Lambda_n}$.
By part (ii) of Theorem~5.17 of~\cite{NSY} it follows that $\Psi^{(n)(s)} \in \caB_{\tilde{F}}([0,1])$, and moreover that $\Psi^{(n)(s)}$ converges to $\Psi^{(s)}$ in $F$-norm with respect to $\tilde{F}$.
Theorem 5.13 of~\cite{NSY} implies
\begin{align}
\sup_{n}\int_{0}^{1}\lV \Psi^{(n)(s)}\rV_{\tilde F}(t) dt<\infty,
\end{align}
see also~\cite[eq. (5.101)]{NSY}.
Therefore, from Theorem 3.8 of \cite{NSY}, we obtain~\eqref{convconv}.
\end{proof}

\section{Factorization of quasi-local automorphisms}\label{sec:stable}
In this section we give our main technical result.
In particular, we study conditions under which a quasi-local automorphism $\tau_{1,0}^\Phi$ ``factorizes'' as in Proposition~\ref{prop:splitstable}, in particular equation~\eqref{eq:factorize}.
In the next theorem we give a sufficient condition in terms of the regions involved and the $F$-function for $\Phi$.

Before we state the full conditions and prove the result, let us briefly outline the main steps.
The idea behind the proof is to compare the full dynamics generated by the interaction $\Phi$ with the ``decoupled'' dynamics $\Phi^{(0)}$.
The latter simply omits all interaction terms of $\Phi$ crossing the boundary of $\Gamma_2 \setminus \Gamma_1$.
The first step is to show that the difference between the dynamics, $\tau^{\Phi}_{1,0} \circ \left(\tau^{\Phi^{(0)}}_{1,0}\right)^{-1}$ is quasi-local, and generated by an interaction as in Theorem~\ref{tsan}.
In the second step we show that this interaction can be well approximated by interaction terms localized in $\Gamma_2' \setminus \Gamma_1'$, with $\Gamma_1' \subset \Gamma_1 \subset \Gamma_2 \subset \Gamma_2'$, in the sense that the contributions \emph{outside} this region sum up to a bounded operator in $\calA_\Gamma$.
In Step 3 this is then used to show that the difference of the full and decoupled dynamics can be written as an automorphism of $\calA_{\Gamma_2'\setminus \Gamma_1'}$ followed by conjugation with a unitary.
This ultimately allows us to write the interaction in form that allows us to apply Proposition~\ref{prop:splitstable}, and provide natural examples of quasi-factorizable automorphisms.

\begin{theorem}
	\label{thm:quasiauto}
Let $(\Gamma, d)$ be a countable $\nu$-regular metric space with constant $\kappa$ as in~\eqref{nreg}.
Let $F$ be an $F$-function on $(\Gamma, d)$ such that the function $G_F$ defined by (\ref{gfdef}) satisfies (\ref{as:galp}) 
for some $\alpha\in(0,1)$.
Suppose that there is an $F$-function $\tilde F$ satisfying (\ref{as:gf}) for this $F$.
Let $\caA_\Gamma$ be a quantum spin system given by (\ref{ag}) and (\ref{agg}).

Let $\Phi\in \caB_{F}([0,1])$ be a path of interactions satisfying $\Phi_1\in \caB_F([0,1])$.
(Recall from definition (\ref{pm}) that this means that $X \mapsto |X| \Phi(X;t)$ is in $\caB_F([0,1])$).
Let
\begin{align}
\Gamma_1'\subset  \Gamma_{1}\subset \Gamma_{2}\subset \Gamma_2'\subset \Gamma.
 \end{align}
For $m\in\nan\cup \{0\}$, $x,y\in\Gamma$, set
\begin{align}
	\label{eq:defnf}
f(m,x,y):=\sum_{X\ni x,y, d\lmk \lmk \Gamma_{2}'\setminus \Gamma_{1}'\rmk^{c}, X\rmk\le m} |X| \sup_{t\in[0,1]} \lV\Phi(X,t)\rV.
\end{align}
We assume that
 \begin{align}\label{anan}
\lmk  \sum_{x\in \Gamma_1} \sum_{y\in \Gamma_2^{c}}+  \sum_{x\in \Gamma_2\setminus \Gamma_1} 
 \sum_{y\in \lmk \Gamma_2\setminus \Gamma_1\rmk^c}\rmk
 \sum_{m=0}^\infty
  G_F(m)f(m,x,y)<\infty
 \end{align}
 Define $\Phi^{(0)}\in \caB_{F}([0,1])$ by
 \begin{align}
 \Phi^{(0)}\lmk X; t\rmk:=
 \left\{
\begin{gathered}
 \Phi\lmk X; t\rmk,\quad \text{if}\; X\subset \Gamma_{1}\; \text{or}\; X\subset \Gamma_{2}\setminus\Gamma_{1}\;  \text{or} \; X\subset \Gamma_{2}^{c}\\
 0,\quad \text{otherwise}
 \end{gathered}
 \right.,
 \end{align}
 for each $X\in\caP_{0}(\Gamma)$, $t\in[0,1]$.
 Then there is an automorphism
$\beta_{\Gamma_2'\setminus \Gamma_1'}$ on $\caA_{\Gamma_2'\setminus \Gamma_1'}$
and a unitary $u\in\caA_{\Gamma}$ such that
\begin{align}\label{eq:quasifactor}
\tau_{1,0}^\Phi=\Ad(u)\circ
\tau_{1,0}^{\Phi^{(0)}}\circ\lmk\tilde  \beta_{\Gamma_2'\setminus \Gamma_1'}\rmk.
\end{align}
\end{theorem}
\begin{proof}
{\it Step 1.}
First we would like to represent $\tau_{1,0}^{ \Phi}\circ \lmk \tau_{1,0}^{\Phi^{(0)}}\rmk^{-1}$
as some quasi-local automorphism, applying Theorem \ref{tsan}.
 Let $\{\Lambda_{n}\}_{n=1}^{\infty}\subset\caP_{0}\lmk \Gamma\rmk$ be an increasing sequence $\Lambda_{n}\nearrow\Gamma$.
 We also define $\Phi^{(1)}\in \caB_{F}([0,1])$ by
 \begin{align}
 \Phi^{(1)}\lmk X; t\rmk:=\Phi^{(0)}\lmk X; t\rmk-\Phi\lmk X; t\rmk,
 \end{align}
 for each $X\in\caP_{0}(\Gamma)$, $t\in[0,1]$.
 
Let $t,s\in[0,1]$.
We apply Theorem~\ref{tsan} to $\Phi^{(1)}$.
Hence we set
\begin{align}
\Psi^{(s)}\lmk
Z, t
\rmk:=
\sum_{m\ge 0} \sum_{X\subset Z,\; X(m)=Z}
\Delta_{X(m)}\lmk
\tau_{t,s}^{\Phi}\lmk
 \Phi^{(1)}\lmk X; t\rmk
\rmk
\rmk
\end{align}
and
\begin{align}
\Psi^{(n)(s)}\lmk
Z, t
\rmk:=
\sum_{m\ge 0} \sum_{X\subset Z, X(m)\cap\Lambda_{n}=Z}
\Delta_{X(m)}\lmk
\tau_{t,s}^{(\Lambda_n)\Phi}\lmk
 \Phi^{{(1)}}\lmk X; t\rmk
\rmk
\rmk.
\end{align}
Corresponding to (\ref{psio}), we obtain
\begin{align}
\tau_{t,s}^{(\Lambda_n),\Phi} \lmk
H_{\Lambda_n,\Phi^{(1)}}
\rmk
= H_{\Lambda_n,\Psi^{(n)(s)}}(t).
\end{align}
Applying Theorem~\ref{tsan},
we have $\Psi^{(n)(s)}, \Psi^{(s)}\in \caB_{\tilde F}([0,1])$, and
\begin{align}\label{convconv1}
\lim_{n\to\infty}\lV
\tau_{t,u}^{\Psi^{(n)(s)}}\lmk A\rmk
-\tau_{t,u}^{\Psi^{(s)}}\lmk A\rmk
\rV=0,\quad A\in\caA_{\Gamma},\quad t,u\in [0,1]
\end{align}
holds.
Note that
\begin{align}
\begin{split}
\frac{d}{dt} \hat\tau_{t,s}^{(\Lambda_n), \Psi^{(n)(s)}}(A)
&=-i\left[ H_{\Lambda_n, \Psi^{(n)(s)}}(t), \hat\tau_{t,s}^{(\Lambda_n), \Psi^{(n)(s)}}(A)
\right] \\
&=-i \left[
\tau_{t,s}^{(\Lambda_n),\Phi} \lmk
H_{\Lambda_n, \Phi^{(1)}}
\rmk,
\hat\tau_{t,s}^{(\Lambda_n), \Psi^{(n)(s)}}(A)
\right].
\end{split}
\end{align}
On the other hand, we have
\begin{align}
\begin{split}
\frac{d}{dt} \tau_{t,s}^{(\Lambda_n), \Phi}&\circ \lmk \tau_{t,s}^{(\Lambda_n),\Phi^{(0)}}\rmk^{-1}(A) \\
&=\tau_{t,s}^{(\Lambda_n),\Phi}
\lmk
i\left[
H_{\Lambda_n,\Phi}(t)-H_{\Lambda_n,\Phi^{(0)}}(t), \lmk \tau_{t,s}^{(\Lambda_n),\Phi^{(0)}}\rmk^{-1}(A)
\right]
\rmk \\
&=-i\left[
\tau_{t,s}^{(\Lambda_n),\Phi} \lmk
H_{\Lambda_n, \Phi^{(1)}}
\rmk,
\tau_{t,s}^{(\Lambda_n), \Phi}\circ \lmk \tau_{t,s}^{(\Lambda_n),\Phi^{(0)}}\rmk^{-1}(A)
\right].
\end{split}
\end{align}
Hence $\hat\tau_{t,s}^{(\Lambda_n), \Psi^{(n)(s)}}(A)$ and $ \tau_{t,s}^{(\Lambda_n), \Phi}\circ \lmk \tau_{t,s}^{(\Lambda_n),\Phi^{(0)}}\rmk^{-1}(A)$
satisfy the same differential equation with the 
$\hat\tau_{s,s}^{(\Lambda_n), \Psi^{(n)(s)}}(A)=\tau_{s,s}^{(\Lambda_n), \Phi}\circ \lmk \tau_{s,s}^{(\Lambda_n),\Phi^{(0)}}\rmk^{-1}(A)=A$.
Therefore we obtain
\begin{align}\label{ata}
\hat\tau_{t,s}^{(\Lambda_n), \Psi^{(n)(s)}}(A)= \tau_{t,s}^{(\Lambda_n), \Phi}\circ \lmk \tau_{t,s}^{(\Lambda_n),\Phi^{(0)}}\rmk^{-1}(A),\quad t\in [0,1],\quad A\in\caA_{\Gamma}.
\end{align}
From the fact that $\hat\tau_{t,u}^{\Psi^{(n)(s)}}\lmk A\rmk=\hat \tau_{t,u}^{(\Lambda_n), \Psi^{(n)(s)}}=\tau_{u,t}^{(\Lambda_n), \Psi^{(n)(s)}}=\tau_{u,t}^{\Psi^{(n)(s)}}$ converges strongly to an automorphism $\tau_{u,t}^{\Psi^{(s)}}$ on 
$\caA_{\Gamma}$ (\ref{convconv1}), 
we have
\begin{align}\label{convconvh}
\lim_{n\to\infty}\lV
\hat\tau_{t,s}^{\Psi^{(n)(s)}}\lmk A\rmk
-\tau_{s,t}^{\Psi^{(s)}}\lmk A\rmk
\rV=0,\quad A\in\caA_{\Gamma}.
\end{align}
On the other hand, by Theorem~\ref{tni},  
we have for $t \in [0,1]$ and $A \in \caA_\Gamma$
\begin{align}
\lim_{n\to\infty}\lV
 \tau_{t,s}^{(\Lambda_n), \Phi}\circ \lmk \tau_{t,s}^{(\Lambda_n),\Phi^{(0)}}\rmk^{-1}(A)
 -\tau_{t,s}^{ \Phi}\circ \lmk \tau_{t,s}^{\Phi^{(0)}}\rmk^{-1}(A)
 \rV=0.
\end{align}
Therefore, taking $n\to\infty$ limit in (\ref{ata}), we obtain
\begin{align}
\tau_{s,t}^{ \Psi^{(s)}}(A)= \tau_{t,s}^{\Phi}\circ \lmk \tau_{t,s}^{\Phi^{(0)}}\rmk^{-1}(A),\quad t,s\in [0,1],\quad A\in\caA_{\Gamma}.
\end{align}
Hence we have
\begin{align}
\tau_{s,t}^{\Phi}=\lmk \tau_{t,s}^{\Phi}\rmk^{-1}
=\lmk \tau_{t,s}^{\Phi^{(0)}}\rmk^{-1}\lmk\tau_{s,t}^{ \Psi^{(s)}}\rmk^{-1}=
\tau_{s,t}^{\Phi^{(0)}}\tau_{t,s}^{ \Psi^{(s)}}
\end{align}
In particular, we get
\begin{align}\label{ttt}
\tau_{1,0}^{\Phi}=
\tau_{1,0}^{\Phi^{(0)}}\tau_{0,1}^{ \Psi^{(1)}}.
\end{align}
{\it Step 2.}
We show that the summation
\begin{align}\label{vt}
V(t):=\sum_{Z\in\pg}
\lmk
\id-
\Pi_{\Gamma_{2}'\setminus \Gamma_{1}'}
\rmk
\lmk
\Psi^{(1)}\lmk Z,t\rmk
\rmk\in\caA_\Gamma
\end{align}
converges absolutely in the norm topology, and uniformly in $t\in[0,1]$.
Set
\begin{align}
V_{n}(t):=\sum_{Z\in\pg,\; Z\subset \Lambda_{n} } 
\lmk
\id-
\Pi_{\Gamma_{2}'\setminus \Gamma_{1}'}
\rmk
\lmk
\Psi^{(1)}\lmk Z,t\rmk
\rmk\in\caA_{\Lambda_{n}},\quad n\in\nan.
\end{align}
From the convergence of (\ref{vt}) uniform in $t$,
we get
\begin{align}\label{kub}
\lim_{n\to\infty}\sup_{t\in[0,1]}\lV V_{n}(t)-V(t)\rV=0.
\end{align}
To prove the convergence of (\ref{vt}), it suffices to prove
\begin{align}\label{smsm}
\lim_{n\to\infty}\sup_{t\in[0,1]}\sum_{Z\in\pg,\; Z\cap \Lambda_{n}^{c}\neq \emptyset } 
\lV
\lmk
\id-
\Pi_{\Gamma_{2}'\setminus \Gamma_{1}'}
\rmk\lmk
\Psi^{(1)}\lmk z,t\rmk
\rmk
\rV
=0.
\end{align}
To prove this, we introduce the following functions. For $m\in\nan\cup \{0\}$, $n\in\nan$,  and $x,y\in\Gamma$, set
\begin{align}
f_{n}(m,x,y):=\sum_{X\ni x,y,  d(X, \Lambda_{n}^{c})\le m\; d\lmk \lmk \Gamma_{2}'\setminus \Gamma_{1}'\rmk^{c}, X\rmk\le m} |X| \sup_{t\in[0,1]} \lV\Phi(X,t)\rV.
\end{align}
Note that $f_n(m,x,y)$ is bounded by $f$ point-wise (by definition) and converges to zero point-wise, by~\eqref{anan}.
Hence by~\eqref{anan} and Lebesgue's dominated convergence theorem, we obtain
\begin{align}\label{bnbn}
\lim_{n\to\infty}
\lmk
\lmk  \sum_{x\in \Gamma_1} \sum_{y\in \Gamma_2^{c}}+  \sum_{x\in \Gamma_2\setminus \Gamma_1} 
 \sum_{y\in \lmk \Gamma_2\setminus \Gamma_1\rmk^c}\rmk
 \sum_{m=0}^\infty
  G_F(m)f_{n}(m,x,y)\rmk=0.
 \end{align}
We have
\begin{align}
&\sup_{t\in[0,1]}\sum_{Z\in\pg,\; Z\cap \Lambda_{n}^{c}\neq \emptyset } 
\lV
\lmk
\id-
\Pi_{\Gamma_{2}'\setminus \Gamma_{1}'}
\rmk
\lmk
\Psi^{(1)}\lmk Z,t\rmk
\rmk\rV
\\
\begin{split}
&\le
\sum_{Z\in\pg,\; Z\cap \Lambda_{n}^{c}\neq \emptyset } 
\sum_{m\ge 0} \sum_{X\subset Z,\; X(m)=Z} \\
&\quad\quad\quad
\left[
\sup_{t\in[0,1]}\lV
\lmk
\id-\Pi_{\Gamma_{2}'\setminus \Gamma_1'}
\rmk
\Delta_{X(m)}\lmk
\tau_{t,1}^{\Phi}\lmk
 \Phi^{(1)}\lmk X; t\rmk
\rmk
\rmk
\rV
\right]
\end{split}\\
&\le
\sum_{m\ge 0} \sum_{X \in \caP_0(\Gamma) \; X(m)\cap \Lambda_{n}^{c}\neq \emptyset }
\sup_{t\in[0,1]}\lV\lmk
\id-
\Pi_{\Gamma_{2}'\setminus \Gamma_{1}'}
\rmk
\Delta_{X(m)}\lmk
\tau_{t,1}^{\Phi}\lmk
 \Phi^{(1)}\lmk X; t\rmk
\rmk
\rmk
\rV\\
&\le2\sum_{m\ge 0} \sum_{X\in\pg\; X(m)\cap \Lambda_{n}^{c}\neq \emptyset , X(m)\cap\lmk \Gamma_{2}'\setminus \Gamma_{1}'\rmk^{c}\neq\emptyset}
\sup_{t\in[0,1]}\lV\Delta_{X(m)}\lmk
\tau_{t,1}^{\Phi}\lmk
 \Phi^{(1)}\lmk X; t\rmk
\rmk
\rmk
\!\rV\\
\begin{split}
& \le2\sum_{m\ge 0} \sum_{X\in\pg\; X(m)\cap \Lambda_{n}^{c}\neq \emptyset , X(m)\cap\lmk \Gamma_{2}'\setminus \Gamma_{1}'\rmk^{c}\neq\emptyset}\\
&\quad\quad\quad\quad\quad\quad\quad
\left[
\sup_{t\in[0,1]}
\frac{4\lV \Phi^{(1)}\lmk X; t\rmk \rV}{C_{F}}\lmk e^{2I(\Phi)}-1\rmk\lv X\rv G_{F}\lmk m\rmk
\right]
\end{split}
\\
\begin{split}
&=\frac{8}{C_{F}}\lmk e^{2I(\Phi)}-1\rmk
\sum_{m\ge 0} \sum_{X\in\pg\; X(m)\cap \Lambda_{n}^{c}\neq \emptyset , X(m)\cap\lmk \Gamma_{2}'\setminus \Gamma_{1}'\rmk^{c}\neq\emptyset }
\\
&\quad\quad\quad\quad\quad\quad\quad
\left[
\sup_{t\in[0,1]}\lmk \lV \Phi^{(1)}\lmk X; t\rmk \rV\rmk
\lv X\rv G_{F}\lmk m\rmk
\right]
\end{split}
\label{remi}
\end{align}
For the fourth inequality, we used Theorem~\ref{tni}~(iii).
From the definition of $ \Phi^{(1)}$, we have
$
 \Phi^{(1)}\lmk X; t\rmk =0, 
$
unless
$X$ has a non-empty intersection with at least two of
$\Gamma_{1}$, $\Gamma_{2}^c$, $\Gamma_{2}\setminus \Gamma_{1}$.
In particular, we have
$
 \Phi^{(1)}\lmk X; t\rmk =0, 
$
unless
$X\cap \Gamma_{1}\neq\emptyset, X\cap \Gamma_{2}^{c}\neq\emptyset$ or
$X\cap\lmk  \Gamma_{2}\setminus \Gamma_{1}\rmk\neq\emptyset,
X\cap\lmk  \Gamma_{2}\setminus \Gamma_{1}\rmk^{c}\neq\emptyset$.
Therefore, if $
 \Phi^{(1)}\lmk X; t\rmk \neq 0, 
$ there should be $x\in\Gamma_{1}$,
$y\in\Gamma_{2}^{c}$ with $X\ni x, y$
or 
$x\in\Gamma_{2}\setminus \Gamma_{1}$
$y\in\lmk  \Gamma_{2}\setminus \Gamma_{1}\rmk^{c}$
with $X\ni x, y$.
We also note that if
$X(m)\cap \Lambda_{n}^{c}\neq \emptyset$ and $X(m)\cap\lmk \Gamma_{2}'\setminus \Gamma_{1}'\rmk^{c}$,
then we have $d(X, \Lambda_{n}^{c})\le m$ and
$d(X,\lmk \Gamma_{2}'\setminus \Gamma_{1}'\rmk^{c})\le m$.
Therefore we have
\begin{align}
&(\ref{remi})
\le
\frac{8}{C_{F}}\lmk e^{2I(\Phi)}-1\rmk
\lmk
\sum_{x\in\Gamma_{1}}\sum_{{y\in\Gamma_{2}^{c}}}
+\sum_{x\in\Gamma_{2}\setminus \Gamma_{1}}\sum_{y\in\lmk  \Gamma_{2}\setminus \Gamma_{1}\rmk^{c}}
\rmk\\
&\sum_{m\ge 0} \sum_{X\in\pg\; d(X, \Lambda_{n}^{c})\le m,\;
d(X,\lmk \Gamma_{2}'\setminus \Gamma_{1}'\rmk^{c})\le m,\; X\ni x,y}
\sup_{t\in[0,1]}\lmk \lV \Phi^{(1)}\lmk X; t\rmk \rV\rmk
\lv X\rv G_{F}\lmk m\rmk\\
&=\frac{8}{C_{F}}\lmk e^{2I(\Phi)}-1\rmk
\lmk
\sum_{x\in\Gamma_{1}}\sum_{{y\in\Gamma_{2}^{c}}}
+\sum_{x\in\Gamma_{2}\setminus \Gamma_{1}}\sum_{y\in\lmk  \Gamma_{2}\setminus \Gamma_{1}\rmk^{c}}
\rmk
\sum_{m\ge 0} f_{n}(m,x,y)G_{F}\lmk m\rmk.
\end{align}
The last part converges to $0$ as $n\to\infty$ because of (\ref{bnbn}).
This proves (\ref{smsm}), and hence that~\eqref{vt} converges.
\\
{\it Step 3.}
Next we decompose $\Psi^{(1)}$ into a $\Gamma_{2}'\setminus \Gamma_{1}'$-part
\begin{align}
\tilde\Psi(Z,t):=\Pi_{\Gamma_{2}'\setminus \Gamma_{1}'}
\lmk
\Psi^{(1)}(Z,t)
\rmk
\end{align}
 and the rest.
 Clearly, we have $\tilde \Psi\in \caB_{\tilde F}([0,1])$.
 Note that
 \begin{align}\label{hvh}
 H_{\Lambda_{n}, \tilde\Psi}(t)+V_{n}(t)
 = H_{\Lambda_{n}, \Psi^{(1)}}(t).
 \end{align}
 
 As a uniform limit of $[0,1]\ni t\mapsto V_{n}(t)\in \caA_{\Gamma}$,
$[0,1]\ni t\mapsto V(t)\in \caA_{\Gamma}$ is norm-continuous.
Because of $\tilde \Psi\in \caB_{\tilde F}([0,1])$,
$[0,1]\ni t\mapsto \tau_{t,s}^{\tilde\Psi}\lmk V(t)\rmk\in \caA_{\Gamma}$ is also norm-continuous,
for each $s\in[0,1]$.
Therefore, for each $s\in [0,1]$, there is a unique norm-differentiable map from
$[0,1]$ to $ \caU\lmk \caA_{\Gamma}\rmk$
such that
\begin{align}
\frac{d}{dt} W^{(s)}(t)=-i \tau_{t,s}^{\tilde\Psi}\lmk V(t)\rmk W^{(s)}(t),\quad
W^{(s)}(s)=\unit.
\end{align}
The solution is given by
\begin{align}\label{wsexp}
W^{(s)}(t)
:=\sum_{k=0}^{\infty }(-i)^{k}
\int_{s}^{t}ds_{1}\int_{s}^{s_{1}}ds_{2}\cdots \int_{s}^{s_{k-1}}ds_{k}
\tau_{s_{1},s}^{\tilde\Psi}\lmk V(s_{1})\rmk
\cdots \tau_{s_{k},s}^{\tilde\Psi}\lmk V(s_{k})\rmk.
\end{align}
Analogously, for each $s\in[0,1]$ and $n\in\nan$, we define
a unique norm-differentiable map from
$[0,1]$ to $ \caU\lmk \caA_{\Gamma}\rmk$
such that
\begin{align}
\frac{d}{dt} W_{n}^{(s)}(t)=-i \tau_{t,s}^{(\Lambda_{n})\tilde\Psi}\lmk V_{n}(t)\rmk W_{n}^{(s)}(t),\quad
W_{n}^{(s)}(s)=\unit.
\end{align}
This differential equation can be solved similarly as in equation~\eqref{wsexp}.
By the uniform convergence (\ref{kub}),
we then have
\begin{align}
\lim_{n}\sup_{t\in[0,1]}\lV W_{n}^{(s)}(t)- W^{(s)}(t) \rV =0.
\end{align}
From this and Theorem \ref{tni} (iv) for $\Psi^{(1)}, \tilde \Psi\in \caB_{\tilde F}([0,1])$,
we have
\begin{align}\label{limlim}
\lim_{n\to\infty}  \tau_{s,t}^{(\Lambda_{n}), \tilde\Psi}\circ \Ad \lmk W_{n}^{(s)}(t)\rmk (A)=
\tau_{s,t}^{ \tilde\Psi}\circ \Ad \lmk W^{(s)}(t)\rmk (A),\\
\lim_{n\to\infty}  \tau_{s,t}^{(\Lambda_{n}), \Psi^{(1)}}(A)=
\tau_{s,t}^{ \Psi^{(1)}}(A),
\end{align}
for any $A\in\caA_{\Gamma}$.

Note that for any $A\in\caA_{\Gamma}$
\begin{align}
&\frac{d}{dt} \tau_{s,t}^{(\Lambda_{n}), \tilde\Psi}\circ \Ad \lmk W_{n}^{(s)}(t)\rmk (A)\\
\begin{split}
&=-i\left[
H_{\Lambda_{n}, \tilde\Psi}(t),
\tau_{s,t}^{(\Lambda_{n}), \tilde\Psi}\circ \Ad \lmk W_{n}^{(s)}(t)\rmk (A)
\right] \\
&\quad\quad\quad\quad\quad\quad\quad
-i\tau_{s,t}^{(\Lambda_{n}), \tilde\Psi}\lmk
\left[
\tau_{t,s}^{(\Lambda_{n}), \tilde\Psi}\lmk V_{n}(t)\rmk,
\Ad \lmk W_{n}^{(s)}(t)\rmk (A)
\right]
\rmk
\end{split}\\
&=-i\left[
H_{\Lambda_{n}, \tilde\Psi}(t)+V_{n}(t),
\tau_{s,t}^{(\Lambda_{n}), \tilde\Psi}\circ \Ad \lmk W_{n}^{(s)}(t)\rmk (A)
\right]\nonumber\\
&=-i\left[
H_{\Lambda_{n}, \Psi^{(1)}}(t),
\tau_{s,t}^{(\Lambda_{n}), \tilde\Psi}\circ \Ad \lmk W_{n}^{(s)}(t)\rmk (A)
\right].
\end{align}
We used  (\ref{inv})
for the second equality and (\ref{hvh}) for the third equality.
On the other hand, for any $A\in\caA_{\Gamma}$, we have 
\begin{align}
\frac{d}{dt} \tau_{s,t}^{(\Lambda_{n}), \Psi^{(1)}} (A)=-i\left[
H_{\Lambda_{n}, \Psi^{(1)}}(t),
\tau_{s,t}^{(\Lambda_{n}), \Psi^{(1)}} (A)
\right].
\end{align}
Therefore, $\tau_{s,t}^{(\Lambda_{n}), \tilde\Psi}\circ \Ad \lmk W_{n}^{(s)}(t)\rmk (A)$
and $ \tau_{s,t}^{(\Lambda_{n}), \Psi^{(1)}} (A)$ satisfy the same differential equation.
Also note that we have
\[
    \tau_{s,s}^{(\Lambda_{n}), \tilde\Psi}\circ \Ad \lmk W_{n}^{(s)}(s)\rmk (A)=
 \tau_{s,s}^{(\Lambda_{n}), \Psi^{(1)}} (A)=A.
\]
 Therefore, we get
 \begin{align}
 \tau_{s,t}^{(\Lambda_{n}), \tilde\Psi}\circ \Ad \lmk W_{n}^{(s)}(t)\rmk (A)
 =\tau_{s,t}^{(\Lambda_{n}), \Psi^{(1)}} (A).
 \end{align}
By (\ref{limlim}), we obtain
\begin{align}
 \tau_{s,t}^{\tilde\Psi}\circ \Ad \lmk W^{(s)}(t)\rmk (A)
 =\tau_{s,t}^{\Psi^{(1)}} (A),\quad A\in\caA_{\Gamma}, \; t,s\in[0,1].
 \end{align}
 Taking inverse, we get 
\begin{align}\label{www}
 \Ad \lmk W^{(s)^{*}}(t)\rmk\circ\tau_{t,s }^{\tilde\Psi}  
 =\tau_{t,s}^{\Psi^{(1)}},\; t,s\in[0,1].
 \end{align}
 {\it Step 4.}
 Combining (\ref{ttt}) and (\ref{www}) we have
 \begin{align}
\tau_{1,0}^{\Phi}=
\tau_{1,0}^{\Phi^{(0)}}\tau_{0,1}^{ \Psi^{(1)}}
=\tau_{1,0}^{\Phi^{(0)}}\circ\Ad \lmk \lmk W^{(1)}(0)\rmk^{*}\rmk\circ\tau_{0,1 }^{\tilde\Psi}.
\end{align}
Setting
\begin{align}
\beta_{\Gamma_2'\setminus \Gamma_1'}:=\tau_{0,1 }^{\tilde\Psi},\quad
u:=\tau_{1,0}^{\Phi^{(0)}}\lmk \lmk W^{(1)}(0)\rmk^{*}\rmk
\end{align}
completes the proof.
\end{proof}

\section{Long-range entanglement}\label{sec:lre}
An interesting problem is to find conditions that lead to a trivial superselection structure.
Topological order is associated to ``long-range entanglement'' that cannot be removed by local operations.
This should be contrasted with product states, which are not entangled at all.
Hence one is interested in states that cannot be transformed into product states by such local operations.
The product states are said to be in the topologically trivial phase~\cite{ChenGW}.

The goal of this section is to show that such a topologically trivial state indeed leads to a trivial superselection structure, at least when we restrict to strictly localized sectors as in equation~\eqref{eq:sselect}.
To make this precise, we recall that the equivalence relation defined in terms of finite depth quantum circuits is somewhat too restrictive in the thermodynamic limit, and one has to look at limits of such automorphisms as well.
In addition, we will only require to be able to ``decouple'' a cone-like region.
Because of transportability of the anyons that is assumed, the choice of cone is not important.
We therefore adopt the following definition.

\begin{definition}\label{def:lre}
Let $\calA_\Gamma$ be the quasi-local algebra of a quantum spin system with $\Gamma = \mathbb{Z}^\nu$. We say that a pure state $\omega$ has \emph{long-range entanglement} (LRE) if there is no quasi-factorizable automorphism $\alpha \in \operatorname{Aut}(\calA_\Gamma)$ such that $\omega \circ \alpha$ is a product state with respect to some cone $\Lambda$. 
Here we say that a state is a product state for a cone $\Lambda$ if it is of the form $\omega = \omega_\Lambda \otimes \omega_{\Lambda^c}$, with $\omega_\Lambda$ a state on $\alg{A}_\Lambda$, and similarly for $\omega_{\Lambda^c}$. 
\end{definition}

\begin{remark}
Since the idea is to capture the trivial phase, the set of allowed automorphisms is dictated by the equivalence relation one puts on the ground states.
Our proofs depend on $\alpha$ being quasi-factorizable, which is why we choose this class of automorphisms in our definition of long-range entanglement.
As we show in Section~\ref{sec:approxsplit}, the notion of quasi-factorizable automorphisms includes natural examples of gapped paths of uniformly bounded finite range interactions.
As we show below, any state that is not long-range entangled has a trivial sector structure.
In fact, the sector structure for states in other phases is also preserved under applying quasi-factorizable automorphisms, if one makes the additional assumption of \emph{approximate Haag duality}~\cite{Ogata21a}.
\end{remark}

The condition that $\Gamma = \mathbb{Z}^\nu$ is not essential.
However, in the general case one should define the appropriate analogue of a cone.
This depends on the localization properties of the excitations one would want to consider, but for the definition to be non-trivial a cone should at least have infinitely many sites.

Note that for a state to be \emph{not} long-range entangled, we only require the condition to hold for a \emph{single} cone $\Lambda$.That is, a state is not long-range entangled if we can disentangle the cone $\Lambda$ from its complement.
Typically the states we are interested in have a large degree of `homogeneity', for example because they will be translation invariant.
Moreover, we will be interested in transportable charges, in the sense that we can move a charge localized in a specific cone to \emph{any} with a unitary operator.
Thus typically one expects that if it is possible to decouple a single cone in this situation, one can do it for more cones.
Since we will not actually need that, we restrict to this simpler definition.

In the following we first consider the situation where the pure reference state $\omega_0$ is a product state with respect to a fixed cone $\Lambda$, i.e., $\omega_0 = \omega_\Lambda \otimes \omega_{\Lambda^c}$ for some states $\omega_\Lambda$ and $\omega_{\Lambda^c}$ on $\calA_\Lambda$ and $\calA_{\Lambda^c}$ respectively.
Below we consider general pure states without long-range entanglement.

We first recall the following Lemma (compare with e.g.~\cite{Matsui01,Matsui10}).
\begin{lemma}\label{lem:split}
    Let $\varphi$ be a pure state on $\calA_\Gamma$ and suppose that there is a cone $\Lambda$ such that $\varphi$ is quasi-equivalent to $\varphi_\Lambda \otimes \varphi_{\Lambda^c}$, where $\varphi_\Lambda := \varphi|\calA_\Lambda$. Then $\calR_\Lambda := \pi_\varphi(\alg{A}_\Lambda)''$ is a factor of Type I, and so is $\mathcal{R}_{\Lambda^c}$. Moreover, Haag duality holds: $\calR_\Lambda = \calR_{\Lambda^c}'$.
\end{lemma}
\begin{proof}
Write $(\pi_\varphi, \calH_\varphi, \Omega_\varphi)$ for the GNS representation of $\varphi$.
Because $\varphi$ is pure, $\pi_\varphi(\alg{A}_\Gamma)''$ is a Type I factor. 
Note that $\calR_\Lambda \vee \calR_{\Lambda^c} = \alg{B}(\calH_{\varphi})$.
Here $\calR_\Lambda \vee \calR_{\Lambda^c}$ is the smallest von Neumann algebra containing both $\calR_\Lambda$ and $\calR_{\Lambda^c}$.
Taking the commutant of this equation, and noting that by locality we have that $\calR_\Lambda \subset \calR_{\Lambda^c}'$, one obtains
\[
	\mathcal{R}_\Lambda' \cap \mathcal{R}_\Lambda \subset \mathcal{R}_\Lambda' \cap \mathcal{R}_{\Lambda^c}' = \mathbb{C} I.
\]
Hence $\calR_\Lambda$ is a factor, and so is $\calR_{\Lambda^c}$.

Since $\varphi$ is quasi-equivalent to $\varphi_\Lambda \otimes \varphi_{\Lambda^c}$ it follows that there is a normal isomorphism $\tau: \pi_\varphi(\alg{A}_\Gamma)'' \to \pi_{\varphi_\Lambda}(\alg{A}_\Lambda)''\,\, \overline{\otimes}\,\, \pi_{\varphi_{\Lambda^c}}(\alg{A}_{\Lambda_c})''$.
The notation $\mathcal{N} \overline{\otimes} \mathcal{M}$ denotes the von Neumann-algebraic tensor product, which by definition is the smallest von Neumann algebra containing the algebraic tensor product $\mathcal{N} \odot \mathcal{M}$.
Because the tensor product of two von Neumann algebras is Type I if and only if both factors are Type I, it follows that $\pi_{\varphi_\Lambda}(\alg{A}_\Lambda)''$ must be Type I, and similarly for $\pi_{\varphi_{\Lambda^c}}(\alg{A}_{\Lambda^c})''$.
Finally, since $\calR_\Lambda$ is a factor, every subrepresentation of $\pi_\Lambda := \pi_\varphi | \alg{A}_\Lambda$ is quasi-equivalent to $\pi_{\Lambda}$ itself.
This is true in particular for $\pi_{\varphi_\Lambda}$, and hence $\calR_\Lambda$ must be of Type I as well.
The same is true for $\calR_{\Lambda^c}$.

Finally, since $\calR_{\Lambda}$ is of Type I, there are Hilbert spaces $\calH_1$ and $\calH_2$ and a unitary $U : \calH_{\varphi} \to \calH_1 \otimes \calH_2$, with
\[
	U \calR_\Lambda U^* = \alg{B}(\calH_1) \> \overline{\otimes} \> I, \quad\textrm{ and }\quad	U \calR_{\Lambda^c} U^* \subset I \> \overline{\otimes} \> \alg{B}(\calH_2).
\]
The inclusion follows because $\calR_{\Lambda^c} \subset \calR_\Lambda'$ by locality, and because $(\alg{B}(\calH_1) \overline{\otimes} I)' = I \> \overline{\otimes}\> \alg{B}(\calH_2)$.
Because $\calR_\Lambda$ and $\calR_{\Lambda^c}$ generate $\alg{B}(\calH_{\varphi})$, it follows that in fact it must be an equality.
Therefore $\calR_\Lambda = \calR_{\Lambda^c}'$.
\end{proof}
\begin{remark}
As is shown in the references cited above, the factors being Type~I implies that $\varphi$ is quasi-equivalent to a product state.
However, Haag duality does not necessarily imply the split property.
\end{remark}

This allows us to prove that if the reference is a product state with respect to a cone, there are no non-trivial representations that are both strictly localizable and transportable.
In other words, the superselection structure is trivial.
We will in fact slightly relax the superselection criterion, and only assume that the representations $\pi$ of interest are \emph{quasi-}equivalent to $\pi_0$.
More precisely, we will be interested in representation $\pi$ such that
\begin{equation}
	\label{eq:nsselect}
        \pi_0|\calA_{\Lambda^c} \sim_{q.e.} \pi|\calA_{\Lambda^c},
\end{equation}
for \emph{all} cones $\Lambda$.
This is true in particular when $\pi$ is unitarily equivalent to $n \cdot \pi_0$ when restricted to observables outside a cone.
Here $n \cdot \pi_0$ is the direct sum of $n$ copies of $\pi_0$, as usual.
The reason to allow this relaxation is that such representations can be constructed naturally when considering non-abelian models~\cite{Szlachanyiv93,Naaijkens2015}.
Note that the condition that~\ref{eq:nsselect} should hold for \emph{every} cone $\Lambda$ is very strong, and as we argued above, captures precisely the localization properties one expects from anyons in 2D.
    The fact that it holds for every cone $\Lambda$ often allows us to draw conclusions about all cones from a result for a single, fixed cone (up to quasi-equivalence).

The following proof is inspired by Proposition~4.2 of~\cite{Mueger99}.
\begin{theorem}\label{thm:trivial}
Let $\omega_0$ be a pure state such that its GNS representation $\pi_0$ is quasi-equivalent to $\pi_\Lambda \otimes \pi_{\Lambda^c}$, with $\pi_\Lambda$ and $\pi_{\Lambda^c}$ irreducible representations of $\calA_\Lambda$ and $\calA_{\Lambda^c}$ respectively.
Consider $\omega_0$ to be the reference state in the superselection criterion.
Then the corresponding sector theory is trivial, in the sense that each representation $\pi$ satisfying the selection criterion~\eqref{eq:nsselect} is
quasi-equivalent to $\pi_0$.
In particular, if $\pi$ is irreducible, then $\pi$ and $\pi_0$ are equivalent.
\end{theorem}
\begin{proof}
Because $\pi\vert_{\caA_{\Lambda^c}}$ is quasi-equivalent to 
$\pi_0\vert_{\caA_{\Lambda^c}}$, which is quasi-equivalent to $\pi_{\Lambda^c}$,
and $\pi_{\Lambda^c}$ is irreducible,
there is a Hilbert space $\caK$  and a unitary $W:\caH\to \caH_{\Lambda^c}\otimes\caK$
such that 
\begin{align}
W \pi(B) W^*=\pi_{\Lambda^c}(B)\otimes\unit_{\caK},\quad B\in \caA_{\Lambda^c}.
\end{align}
    Because $\pi_{\Lambda^c}(\caA_{\Lambda^c})''$ is a Type I factor, it follows that
\[
    \left(\pi_{\Lambda^c}(\alg{A}_{\Lambda^c}) \otimes \unit_{\caK}\right)' = \unit_{\caH_{\Lambda_c}} \otimes \caB(\caK).
\]
By the commutativity of $\caA_{\Lambda}$ and $\caA_{\Lambda^c}$, it follows that for all $A \in \caA_{\Lambda}$, we have that 
$W \pi(A) W^* \in  \left(\pi_{\Lambda^c}(\alg{A}_{\Lambda^c}) \otimes \unit_{\caK}\right)' $.
Thus we see that there is a representation
$\rho$ of $\caA_{\Lambda}$ on $\caK$
such that
\begin{align}\label{wpw}
W\pi(A)W^*=\unit_{\caH_{\Lambda^c}}\otimes \rho(A),\quad
A\in \caA_{\Lambda}.
\end{align}
Consider a cone $\Lambda'$ such that $\Lambda \subset (\Lambda')^c$.
Then, by applying the superselection criterion and restricting to the cone $\Lambda$, it follows that
the representation $\pi\vert_{\caA_{ \Lambda}}$ is quasi-equivalent to 
$\pi_0\vert_{\caA_{\Lambda}}$, which in turn is quasi-equivalent to the irreducible representation
$\pi_{\Lambda}$.
On the other hand, from equation~\eqref{wpw}, 
$\rho$ is quasi-equivalent to  $\pi\vert_{\caA_{\Lambda}}$.
Hence $\rho$  is quasi-equivalent to the irreducible $\pi_{\Lambda}$.
Therefore, there are a Hilbert space $\caK_1$ and a unitary $V: \caK\to \caH_\Lambda\otimes \caK_1$
such that
\begin{align}
V\rho(A)V^*=\pi_{\Lambda}(A)\otimes \unit_{\caK_1},\quad A\in\caA_{\Lambda}.
\end{align}
Hence we get
\begin{align}
\lmk \unit_{\Lambda^c}\otimes V\rmk W\pi(AB) W^* \lmk \unit_{\Lambda^c}\otimes V\rmk^*
=\pi_{\Lambda^c}(B)\otimes\pi_{\Lambda}(A)\otimes\unit_{\caK_1}
\end{align}
for all $ A\in\caA_{\Lambda}$ and $B\in\caA_{\Lambda^c}$.
As the right hand side is quasi-equivalent to $\pi_0$, 
$\pi$ is quasi-equivalent to $\pi_0$.
\end{proof}

\begin{remark}
Note that the assumption in the theorem is a 2D analogue of the split property for 1D spin chains.
It should be noted that it does \emph{not} hold for models such as the toric code, which have non-trivial excitations (or sectors) localized in cones.
The reason is that the ground state has long-range entanglement and cannot be converted into a product state with local operations.
However, as we already mentioned in the introduction, we still have the \emph{approximate} or \emph{distal} split property: a Type I factor $\mathcal{R}_{\Lambda_1} \subset F \subset \mathcal{R}_{\Lambda_2^c}'$ exists if the boundary of the cones $\Lambda_1 \subset \Lambda_2$ are sufficiently distant \cite{FiedlerN}.
What is ``sufficiently distant'' depends on the model, as mentioned in the introduction.
    In general, for example if we perturb with a quasi-local automorphism with a non-zero Lieb-Robinson bound, we need to have that the cone $\Lambda_2$ has a wider opening angle than $\Lambda_1$ as well.
In any case, if the (strict) split property does not hold, it is no longer possible to decompose the representation as a tensor product of representations of $\alg{A}_{\Lambda}$ and $\alg{A}_{\Lambda^c}$.
\end{remark}

The theorem says that, as expected, the product state does not have any non-trivial superselection sectors.
For a general state without long range entanglement, we can try use the quasi-local automorphism $\alpha$ from Definition~\ref{def:lre} to relate the sectors of $\omega \circ \alpha$ with those of $\omega$.
In general there is no reason why $\omega$ should be quasi-equivalent to $\omega \circ \alpha$, so it does not follow directly that $\omega \circ \alpha$ has trivial sectors.
    However if $\alpha$ comes from quasi-local dynamics satisfying Theorem~\ref{thm:quasiauto}, we can relate the sectors of $\pi_\omega$ and $\pi_\omega \circ \alpha$.
    The key point is that we can almost ``factorize'' the automorphism $\alpha$ into automorphisms acting on a cone $\Lambda$ and its complement, up to conjugation with a unitary in $\calA_\Gamma$ and an automorphism acting non-trivially only near the border of $\Lambda$.
More precisely, we will consider $\alpha$ that are quasi-factorizable in the sense of Definition~\ref{def:quasifactor}.
In Section~\ref{sec:approxsplit} we will show how such automorphisms can be obtained using Theorem~\ref{thm:quasiauto}.

\begin{theorem}\label{thm:invariant}
Let $(\caH_0,\pi_0)$ be a representation.
{Let $\alpha$ be a quasi-local automorphism such that for every cone $\Lambda$, we can find an inclusion of cones $\Gamma_1 \subset \Lambda \subset \Gamma_2$ such that $\alpha$ is quasi-factorizable with repsect to this inclusion.}
Suppose that a representation $\pi$ satisfies the superselection criterion
for $\pi_0$ in the sense that
for all cones $\Lambda$ in $\mathbb{Z}^2$, we have
\begin{align}
\pi\vert_{\caA_{\Lambda^c}}\sim_{q.e.} \pi_0\vert_{\caA_{\Lambda^c}}.
\end{align}
Then  $\pi \circ \alpha$  satisfies the superselection criterion~\eqref{eq:nsselect}
for $\pi_0 \circ \alpha$
\end{theorem}
\begin{proof}
Let $\Lambda$ be a cone.
We will show that 
\begin{align}
\pi\circ\alpha\vert_{\caA_{\Lambda^c}}\sim_{q.e.} \pi_0\circ\alpha\vert_{\caA_{\Lambda^c}}.
\end{align}
By assumption we can factorize $\alpha$ as
\begin{align}
    \alpha=\Ad(u)\circ\widetilde{\Xi} \circ \lmk\alpha_{\Lambda}\otimes \alpha_{\Lambda^c}\rmk,
\end{align}
as in Definition~\ref{def:quasifactor}.
From this, for any $A\in\caA_{\Lambda^c}$, we have
\begin{equation}
\begin{split}
    \pi\circ\alpha(A) &=\pi\circ\Ad(u)\circ\widetilde{\Xi}(\alpha_{\Lambda^c}(A))  \\
                      &=\Ad\lmk\pi(u)\rmk\circ \pi \circ \widetilde{\Xi}(\alpha_{\Lambda^c}(A)) \\
                      &=\Ad\lmk\pi(u)\rmk\circ \pi \vert_{\caA_{\Gamma_1^c}}\circ \widetilde{\Xi}(\alpha_{\Lambda^c}(A)).
\end{split}
\end{equation}
This implies
\begin{align}\label{aaa}
    \pi\circ\alpha\vert_{\caA_{\Lambda^c}}
    \sim_{q.e.} \pi \vert_{\caA_{\Gamma_1^c}}\circ \widetilde{\Xi}\circ\alpha_{\Lambda^c}\vert_{\caA_{\Lambda^c}}.
\end{align}
(In fact this is even a unitary equivalence).
Similarly, we have
\begin{align}\label{bbb}
    \pi_0\circ\alpha\vert_{\caA_{\Lambda^c}}
    \sim_{q.e.} \pi_0 \vert_{\caA_{\Gamma_1^c}}\circ \widetilde{\Xi} \circ \alpha_{\Lambda^c}\vert_{\caA_{\Lambda^c}}.
\end{align}
Because we have $\pi\vert_{\caA_{\Gamma_1^c}}\sim_{q.e.} \pi_0\vert_{\caA_{\Gamma_1^c}}$ by virtue of the superselection criterion,
we get
\begin{align}
    \pi \vert_{\caA_{\Gamma_1^c}}\circ \widetilde{\Xi}\circ \alpha_{\Lambda^c}\vert_{\caA_{\Lambda^c}}
    \sim_{q.e.} \pi_0 \vert_{\caA_{\Gamma_1^c}}\circ \widetilde{\Xi} \circ \alpha_{\Lambda^c}\vert_{\caA_{\Lambda^c}}.
\end{align}
Combining this with (\ref{aaa}) and (\ref{bbb}), we get
\begin{align}
    \pi\circ\alpha\vert_{\caA_{\Lambda^c}}\sim_{q.e.}
\pi_0\circ\alpha\vert_{\caA_{\Lambda^c}}.
\end{align}
This proves the claim.
\end{proof}

Combining the two theorems in this section then shows that short-range entangled states indeed have a trivial sector structure.
\begin{corollary}\label{cor:triviality}
Let $(\caH_0,\pi_0)$ be an irreducible representation
which factorizes as  $\pi_0=\pi_{\Lambda}\otimes \pi_{\Lambda^c}$
for some cone $\Lambda$, where $(\pi_{\Lambda},\caH_{\Lambda})$, 
$(\pi_{\Lambda^c},\caH_{\Lambda^c})$
are irreducible representations of $\caA_{\Lambda}$, $\caA_{\Lambda^c}$
respectively.
Let $\alpha$ be a quasi-local automorphism which is quasi-factorizable for all cones $\Lambda$.
Suppose that a representation $\pi$ satisfies the superselection criterion
for $\pi_0 \circ \alpha$ in the sense that
for all cones $\widetilde\Lambda$ in $\mathbb{Z}^2$, we have
\begin{align}
\pi\vert_{\caA_{\widetilde\Lambda^c}}\sim_{q.e.} \pi_0 \circ \alpha\vert_{\caA_{\widetilde\Lambda^c}}.
\end{align}
Then  $\pi$ is quasi-equivalent to $\pi_0 \circ \alpha$. 
In particular, if $\pi$ is irreducible, then $\pi$ and $\pi_0$ are equivalent.
\end{corollary}
\begin{proof}
If $\alpha$ is a quasi-local automorphism, the same is true for $\alpha^{-1}$, and it is quasi-factorizable as well.
Because $\pi$ satisfies the superselection criterion
for $\pi_0 \circ \alpha$ and $\alpha^{-1}$ is a quasi-local automorphism,
by Theorem \ref{thm:invariant},
$\pi\circ\alpha^{-1}$  satisfies the superselection criterion
for $\pi_0\circ\alpha\circ\alpha^{-1}=\pi_0$.
 Then by Theorem \ref{thm:trivial},
$\pi\circ\alpha^{-1}$ is quasi-equivalent to $\pi_0$.
From this, it follows that $\pi$ is quasi-equivalent to $\pi_0 \circ \alpha$. 
\end{proof}

Note that this applies in particular to states which are not long-range entangled according to Definition~\ref{def:lre}.
    Indeed, suppose that $\omega$ is a pure state, and $\alpha$ a quasi-factorizable automorphism such that $\omega \circ \alpha = \omega_\Lambda \otimes \omega_{\Lambda^c}$ for some cone $\Lambda$ and states $\omega_\Lambda$ of $\calA_{\Lambda}$ and $\omega_{\Lambda^c}$ of $\calA_{\Lambda^c}$.
Then $\omega_0 := \omega_\Lambda \otimes \omega_{\Lambda^c}$ is a pure state, and so must be $\omega_\Lambda$ and $\otimes_{\Lambda^c}$, as otherwise we could write $\omega_0$ as a non-trivial convex combination of two distinct states.
But then the GNS representation $\pi_0$ of $\omega_0$ satisfies the assumptions of Corollary~\ref{cor:triviality}.
Since $\pi_0 \circ \alpha$ is a GNS representation for $\omega \circ \alpha$, it follows that $\omega \circ \alpha$ has no non-trivial sectors.

\begin{remark}
    We argued that a state that satisfies the strict split property for a given cone is trivial in the sense that there are no anyonic excitations (superselection sectors).
    It is however possible to further classify this trivial sector, for example if there is an on-site symmetry $G$.
    In that case, it is natural to demand that two states are only in the same gapped phase if they can be connected by a continuous path of gapped Hamiltonians respecting the $G$-symmetry~\cite{ChenGW}.
In two dimensions, the set of states that are in the trivial phase (i.e., containing the product state with respect to each site) can then be classified by a cocycle in $H^3(G, U(1))$~\cite{Ogata21}.
    However, in our definition, the absence of long-range entanglement does not necessarily imply that the state is such a product of single-site states.
    It seems plausible that if we demand the split property to hold for \emph{any} cone, this would follow.
\end{remark}

We conclude this section with a brief discussion. 
Here, we focussed on necessary conditions for the existence of anyons.
While we have showed that long-range entanglement is a necessary condition, it remains an open problem to find \emph{sufficient} conditions.
In particular, there is no guarantee that a state with long-range entanglement has any non-trivial sectors at all (and in fact given the selection criterion~\ref{eq:nsselect} that should generally not be expected if the reference state is far from homogeneous). 
In addition, even if non-trivial sectors do exist, they are not necessarily anyons.
In fact, in three or higher spatial dimensions, cone-localized sectors have bosonic or fermionic statistics (cf.~\cite{BuchholzF}), but in 2D anyons are a possibility, as for example the abelian quantum double models show~\cite{FiedlerN}.
Although there is a technical condition that implies the corresponding category is modular (which in particular implies that all sectors are anyons), the physical interpretation of this criterion is unclear~\cite[Thm. 5.3]{NaaijkensKL}.

We focussed on the trivial phase here, but one can show that if there are \emph{non-}trivial sectors, the full braided tensor category describing the sectors is invariant under quasi-factorizable automorphisms~\cite{Ogata21a}.
This requires that approximate Haag duality holds, a weaker version of Haag duality that can be shown to be stable under quasi-factorizable automorphisms.
There is another natural generalization of the superselection criterion~\eqref{eq:sselect}, which does not require Haag duality, but a variant of the split property instead~\cite{ChaNN18}.
Given that the spectral flow is quasi-local, it is natural to look at representations that can be localized in cones up to some exponentially decaying error.
This leads to the notion of approximately localizable endomorphisms, and one can develop the full sector theory (including e.g. braiding of charges) using them.
These properties are stable upon applying the quasi-local spectral flow.
We should add the caveat that this is a result about \emph{approximately} localized sectors, i.e.\ localized up to some exponentially decaying error, and we cannot rule out that despite the absence of strictly localized sectors, there is a non-trivial \emph{approximately} localized sector.
In abelian quantum double models, this can be ruled out by imposing an ``energy criterion'', essentially excluding any possible confined charges~\cite{ChaNN18}.
We presently do not know if the absence of such sectors can be proven from more fundamental assumptions.
For example, in the case of strict localization it is not necessary.
The results in this section and in \cite{ChaNN18} strongly suggest that in a state with short-range entanglement, there are no approximately localizable sectors either.
 
\section{Approximate split property for cone algebras}\label{sec:approxsplit}
We apply the results of Section~\ref{sec:stable} to two-dimensional models, and give natural examples of quasi-factorizable automorphisms.
In Section~\ref{sec:lre} we have already discussed the split property for a cone and its complement.
As already mentioned, this strong version does not hold for, for example, abelian quantum double models, where only a weaker version is true~\cite{FiedlerN,Naaijkens12}.
This in turn is a key assumption in the stability of superselection sectors analysis in~\cite{ChaNN18}.
Although there we only need the approximate split property for the ``unperturbed'' model, it is interesting to know if it is in fact a property of the whole phase.
Hence, in this section, we show that for suitable perturbations this is indeed the case, and the perturbed model also satisfies the approximate split property.
For simplicity we restrict to 2D systems and finite range interactions, although we expect that with a more careful analysis, the results extend to a wider class of interactions and to systems in three or more spatial dimensions.

Let us recall that if $F$ is an $F$-function, $F_r(r) := e^{-r} F(r)$ is also an $F$-function.
This is an example of a \emph{weighted} $F$-function in the terminology of Ref.~\cite{NSY}.
Such weighted $F$-functions have favorable decay properties, as can be seen in the following Lemma.
\begin{lemma}
Let $(\Gamma,d)$ be $\mathbb{Z}^2$ with the usual metric.
Then there is a $C>0$ such that we have the following estimate for all $m > \sqrt{2}$:
\label{lem:gfdecay}
\begin{equation}
	G_{F_r}(m) \leq C F(m-\sqrt{2}) m e^{-m},
\end{equation}
where $G_{F_r}$ is as defined in equation~\eqref{gfdef}.
\end{lemma}
\begin{proof}
By translation invariance of the metric and $\Gamma$ we do not need the supremum in equation~\eqref{gfdef}.
Hence we get
\begin{align*}
	G_{F_r}(m) &= \sum_{|x| \geq m} e^{-|x|} F(|x|) \\
			& \leq 2 \pi \int_m^\infty r e^{-r+\sqrt{2}} F(r-\sqrt{2})\,dr \\
			& \leq 2 \pi e^{\sqrt{2}} F(m-\sqrt{2}) \int_m^\infty r e^{-r}\,dr \\
			& \leq 4 \pi e^{\sqrt{2}} F(m-\sqrt{2}) m e^{-m}.
\end{align*}
This can be seen by noting that
\begin{equation}
\int_x^{x+1} \int_y^{y+1} e^{-|(x,y)|+\sqrt{2}} F(|(x,y)|) dx dy \geq e^{-|(x,y)|} F(|(x,y)|)
\end{equation}
for $x,y \geq 0$ (since $F$ is positive and decreasing), and doing a coordinate transformation to polar coordinates.
\end{proof}

It is possible to generalize the lemma to other suitable weightings $F_{g}(r) := e^{-g(r)} F(r)$ (see e.g.~\cite{ChaNN18}).
This could be necessary because in applications one would need to assume that interactions have finite interaction norm with respect to the \emph{weighted} $F$-function, instead of $F$ itself.
Since we will consider only bounded range interactions, this is not an issue for us and we restrict to the easier case for simplicity.

\begin{theorem}
\label{thm:conesum}
Let $\Gamma = \mathbb{Z}^2$ with the usual metric $d$ and consider the corresponding quantum spin system $\calA_\Gamma$, where the local dimension of the spins is uniformly bounded.
Let $t \mapsto \Phi(X;t)$ be a path of dynamics such that $\|\Phi(X;t)\|$ is uniformly bounded both in $X$ and $t$.
Moreover assume that $\Phi$ is of bounded range, and let $F$ be an $F$-function. 
Then $\Phi \in \caB_{F_r}([0,1])$, and it generates quasi-local dynamics $\tau^\Phi_{t,s}$.
Assume that $\Gamma_1 \subset \Gamma_2$ is an inclusion of cones such that their borders are sufficiently far away, in the sense that the lines marking the boundaries of the cones are not parallel.
Then there exist cones $\Gamma_1' \subset \Gamma_1$ and $\Gamma_2' \supset \Gamma_2$ such that the conditions of Theorem~\ref{thm:quasiauto} are satisfied.
\end{theorem}
\begin{proof}
\begin{figure}
	\includegraphics[width=\textwidth]{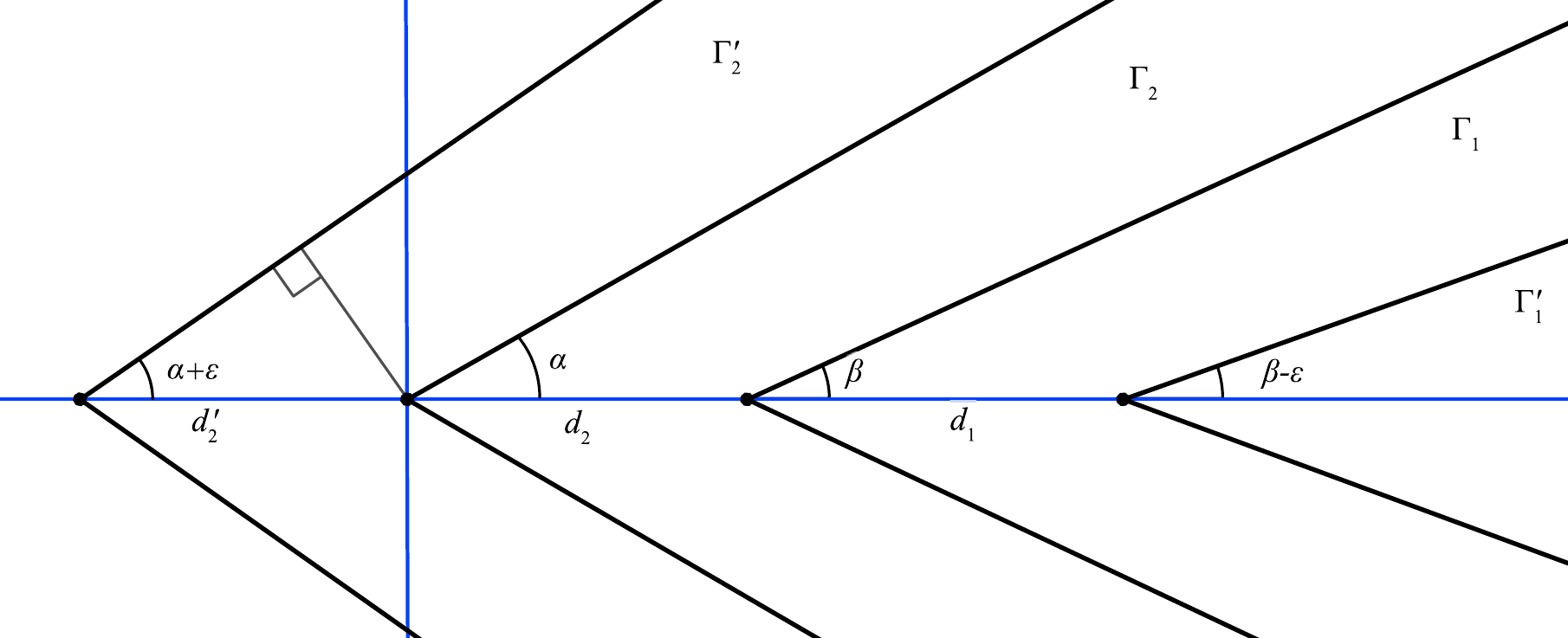}
	\caption{Cones as in Theorem~\ref{thm:conesum}.}
	\label{fig:cones}
\end{figure}
Without loss of generality we may assume that the cones $\Gamma_1$ and $\Gamma_2$ have their center line in the direction of the positive $x$-axis.
We write $\alpha$ for the opening angle of $\Gamma_2$ and $\beta$ for the opening angle $\Gamma_1$ (see Figure~\ref{fig:cones}).
The distance between their tips will be denoted by $d_2$.
Let $0 < \epsilon < \beta$ such that $\alpha + \epsilon < \pi/2$.
We can then choose cones $\Gamma_1'$ and $\Gamma_2'$ as in the figure.
Later in the proof we will provide convenient values for $d_1$ and $d_2'$, but we note that with a little extra work is is possible to show that any positive value will do.

We show that we can apply Theorem~\ref{thm:quasiauto}.
First note that $\Gamma$ is 2-regular, since the number of points in a disk of radius $r$ scales with the area.
Because the interaction range is uniformly bounded and because of 2-regularity, there are constants $C_{\#}$ and $d_\Phi$ such that $\Phi(X;t) = 0$ whenever $|X| > C_{\#}$ or $\diam(X) > d_\Phi$.
It follows that $\Phi_1 \in \caB_{F_r}([0,1])$.
With Lemma~\ref{lem:gfdecay} it is also clear that $G_{F_r}^{\alpha}$ has finite moments for $\alpha \in (0,1]$ (in the sense of equation~\eqref{as:galp}) and we can find a suitable $F$-function $\widetilde{F}$ such that equation~\eqref{as:gf} is satisfied for $F_r$.

It remains to be shown that equation~\eqref{anan} is satisfied.
As a first step we study the function $f(m,x,y)$ of equation~\eqref{eq:defnf}.
Note that the summation in the definition is over certain subsets of $X$ such that $x,y \in X$.
Hence if $d(x,y) > d_\Phi$ we have $\Phi(X;t) = 0$ and consequently $f(m,x,y) = 0$.
Similarly, the summation is only over $X$ such that $d(X, (\Gamma_2' \setminus \Gamma_1')^c) \leq m$.
Hence, it follows that $f(m, x, y) = 0$ unless $d(x, (\Gamma_2'\setminus\Gamma_1')^c) \leq m + d_\Phi$, or the same is true for $y$.
Or giving a rougher estimate, $f(m,x,y) = 0$ unless $d(x, (\Gamma_2'\setminus\Gamma_1')^c) \leq m + 2 d _\Phi$, regardless of $y$.

Now consider the case where $d(x,y) \leq d_\Phi$ and $m$ large enough such that $d(x, (\Gamma_2' \setminus \Gamma_1')^c) \leq m + 2 d_\Phi$.
In that case, we have
\begin{equation}
	\label{eq:fmxybound}
	f(m, x, y) = \sum_{x, y \ni X} |X| \sup_{t} \| \Phi(X; t) \| \leq C_{\#} M 2^{|b_0(d_\Phi)|}, 
\end{equation}
where $M := \sup_X \sup_{t \in [0,1]} \| \Phi(X;t) \|$, which is finite by assumption.
We also used translation invariance of the metric (and $\Gamma$), and that by the finite range assumption any contributing subset $X$ must be contained in $b_x(d_\Phi)$.
There are at most $2^{|b_0(d_\Phi)|}$ of such subsets, leading to the claimed bound.

Next note that Lemma~\ref{lem:gfdecay} gives us the following estimate:
\begin{equation}
	\label{eq:gsumbd}
	\sum_{m = k}^\infty G_{F_r}(m) \leq C F(k-\sqrt{2}) \sum_{m = k}^\infty m e^{-m} \leq C F(0) \frac{e^{-k+1}((e-1)k+1)}{(e-1)^2}
\end{equation}
whenever $k \geq 2$.
Note in particular the factor of $e^{-k+1}$, which will be important to guarantee convergence in our case.

We now return to equation~\eqref{anan}.
Note that $d(\Gamma_1, \Gamma_2^c) = d_2 \sin \alpha$. If this is greater than $d_\Phi$, by the remarks above the first summation (over $x \in \Gamma_1$ and $y \in \Gamma_2^c$) vanishes.
In general, since the cone $\Gamma_2$ has a wider opening angle than $\Gamma_1$, we see that there are only finitely many pairs $x \in \Gamma_1$ and $y \in \Gamma_2^c$ with $d(x,y) \leq  d_\Phi$, and hence only finitely many contributions to the summation.
Together with equations~\eqref{eq:fmxybound} and~\eqref{eq:gsumbd} it can be seen that this contribution is finite.

\begin{figure}
	\includegraphics[width=\textwidth]{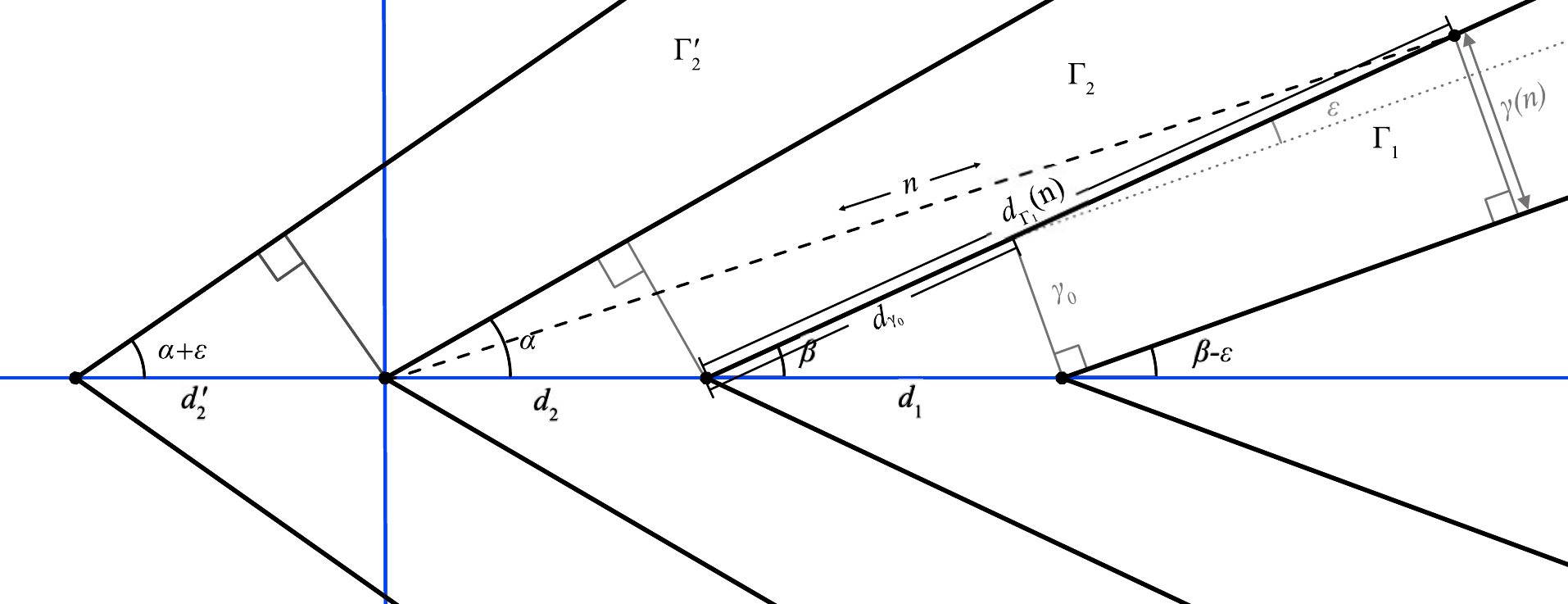}
	\caption{Definition of various distances.}
	\label{fig:distances}
\end{figure}

At this point we are left with estimating the following summation:
\begin{equation}
\sum_{x\in \Gamma_2\setminus \Gamma_1} \lmk \sum_{y\in \Gamma_2^c} + \sum_{y \in \Gamma_1} \rmk
\sum_{m=0}^\infty G_{F_r}(m)f(m,x,y),
\end{equation}
where we have split up the summation over $(\Gamma_2 \setminus \Gamma_1)^c$ into two parts.
We consider the summation over $\Gamma_2^c$, the other one can be handled in the same manner.
Note that $d(\Gamma_2, (\Gamma_2')^c) = d_2' \sin(\alpha+\epsilon)$.
Similarly, $d(\Gamma_2 \cap b_0(n)^c, \Gamma_2') = d_2' \sin(\alpha+\epsilon) + n \sin(\epsilon)$.
Write $d_{\Gamma_1}(n)$ for the distance between the tip of the cone $\Gamma_1$ and the circle of radius $n$ based on the tip of $\Gamma_2$, where we set $d_{\Gamma_1}(n) = 0$ if they do not intersect (see Fig.~\ref{fig:distances} for an idea of the various distances we need to introduce).
In case it is non-zero, we see that in fact
\[
	d_{\Gamma_1}(n) = \sqrt{n^2 - d_2^2(1-\cos^2 \beta)} - d_2 \cos \beta.
\]
Let $\gamma_0$ be the distance from the tip of $\Gamma_1'$ to the intersection of the line perpendicular to the boundary of $\Gamma_1'$ and the boundary of $\Gamma_1$.
We write $d_{\gamma_0}$ for the distance of the tip of $\Gamma_1$ to this intersection.
Then for large enough $n$ the distance of the intersection of the circle of radius $n$ with the boundary of $\Gamma_1$ and the boundary of $\Gamma_1'$ is given by
\[
	\gamma(n) = \gamma_0 + (d_{\Gamma_1}(n)-d_{\gamma_0}) \sin \epsilon.
\]
From the geometric situation we see that $d_{\Gamma_1}(n+k) - d_{\Gamma_1}(n) \geq k$, hence $\gamma(n)$  grows at least linearly in $n$.

Let $n_0$ be the smallest integer such that 
\begin{equation}
	d_0 := \min \{ d_2' \sin(\alpha+\epsilon) + n_0 \sin(\epsilon), \gamma(n_0) \} > 2 d_\Phi.
\end{equation}
Write $B_k := \left(b_0(d_0 + (k+1)/\sin(\epsilon)) \setminus b_0(d_0 + k/\sin(\epsilon)) \right)$.
We now rewrite the summation as
\[
	\left( \sum_{x \in b_0(d_0) \cap (\Gamma_2 \setminus \Gamma_1)} + \sum_{k=0}^\infty 
	\sum_{x \in B_k \cap (\Gamma_2 \setminus \Gamma_1)}  \right) \sum_{y \in \Gamma_2^c} \sum_{m = 0}^\infty G_{F_r}(m) f(m,x,y).
\]
For the first summation over all $x \in \Gamma_2 \setminus \Gamma_1$ in the ball around the origin we note that there are only finitely many such $x$.
We have already seen that for any given $x$, there are only finitely many $y$ (in fact, this number can be bounded from above independently of $x$) such that $f(m,x,y)$ is non-zero.
Again by equations~\eqref{eq:fmxybound} and~\eqref{eq:gsumbd} it follows that the first summation is finite.

For the second summation, note that if $x \in B_k \cap (\Gamma_2 \setminus \Gamma_1)$, then $d(x, (\Gamma_2')^c ) \geq k  + 2 d_\Phi$ and $d(x, \Gamma_1') \geq k + 2 d_\Phi$, and hence $d(x, (\Gamma_2' \setminus \Gamma_1')^c)) \geq k + 2 d_\Phi$.
By what we have seen earlier, this implies that $f(m, x, y) = 0$ if $m < k$ for such $x \in B_k \cap (\Gamma_2 \setminus \Gamma_1)$.
Furthermore, because of the finite range assumption, contributing pairs $x \in B_k$ and $y \in \Gamma_2^c$ must be within a ``band'' of width $d_\Phi$ around each side of the boundary of $\Gamma_2 \setminus \Gamma_1$.
It follows that we can bound the number of pairs $(x,y) \in (B_k \cap \Gamma_2 \setminus \Gamma_1) \times \Gamma_2^c$ by some constant $C_p > 0$ independent of $k$.
Putting this together we can estimate the second summation as follows.
\begin{equation}
	\begin{split}
\sum_{k=0}^\infty \sum_{x \in B_k \cap (\Gamma_2 \setminus \Gamma_1)} & \sum_{y \in \Gamma_2^c} \sum_{m = 0}^\infty G_{F_r}(m) f(m,x,y) \\ 
&\leq \sum_{k=0}^\infty C_p C_{\#} 2^{|b_0(d_\Phi)|} \sum_{m=k}^\infty G_{F_r}(m) \\
&\leq C' \sum_{k=0}^\infty e^{-k+1} ((e-1)k +1) < \infty 
\end{split}
\end{equation}
for some $C' > 0$.
Here we again used the estimates~\eqref{eq:fmxybound} and~\eqref{eq:gsumbd}.
This completes the proof.
\end{proof}

We expect that with a more careful analysis one could allow for more general interactions, as long as they decay sufficiently fast.
It does however seem necessary that that $\Gamma_2'$ has a bigger opening angle than $\Gamma_2$, so that towards infinity the distance between their respective boundaries grows.
This is necessary to ensure that for $x,y$ far from the origin, $f(m,x,y)$ is non-zero only for large $m$.
Together with the decay properties of $G_F$ of Lemma~\ref{lem:gfdecay} this ensures that the sum converges.

The following now follows immediately from the theorem, by using Proposition~\ref{prop:splitstable}.
\begin{corollary}
Let $\calA_\Gamma$ and $t \mapsto \Phi(X;t)$ be as in Theorem~\ref{thm:conesum} and  $\tau^\Phi_{t,s}$ the corresponding quasi-local dynamics.
Assume that $\Gamma_1 \subset \Gamma_2$ is an inclusion of cones such that their borders are sufficiently far away and in the representation $\pi$ of $\calA_\Gamma$ we have the split property with respect to these cones.
Then there exist cones $\Gamma_1' \subset \Gamma_1$ and $\Gamma_2' \supset \Gamma_2$ such that $\pi \circ \tau^\Phi_{1,0}$ satisfies the split property with respect to $\Gamma_1' \subset \Gamma_2'$.
\end{corollary}

Finally, it allows us to construct examples of quasi-factorizable automorphisms.
\begin{corollary}
Let $\alpha = \tau_{0,1}^\Phi$, with $\Phi$ as in Theorem~\ref{thm:conesum}.
Then, for every cone $\Lambda$, we can find cones $\Gamma_1' \subset \Lambda \subset \Gamma_2'$ such that $\alpha$ is quasi-factorizable with respect to this inclusion.
\end{corollary}
\begin{proof}
We will apply Theorem~\ref{thm:quasiauto}; we shall see later why the conditions are satisfied.
Suppose that the cone $\Lambda$ has opening angle $\theta$.
Fix some cone $\Lambda_0$ which has the same apex and central axis as
$\Lambda$ but with a larger angle $\theta_0>\theta$, satisfying
$\Lambda\subset \Lambda_0$.
Set $\Gamma_1:=\Lambda$, $\Gamma_2:=\Lambda_0$. 
Then, by Theorem~\ref{thm:conesum}, there are cones $\Gamma_1' \subset \Gamma_1$ and $\Gamma_2' \supset \Gamma_2$ such that the conditions of Theorem~\ref{thm:quasiauto} are satisfied.
Recall that $\alpha=\tau_{0,1}^{\Phi}$.
Then $\tau_{0,1}^{\Phi^{(0)}}$, in the notation of Theorem~\ref{thm:quasiauto}, decomposes as
\begin{align}
\tau_{0,1}^{\Phi^{(0)}}
=\alpha_{\Gamma_1}\otimes\alpha_{\Gamma_2\setminus \Gamma_1}\otimes
\alpha_{\Gamma_2^c}
\end{align}
where $\alpha_{\Gamma_1} \in \operatorname{Aut}(\calA_{\Gamma_1})$, and similar for the others.
Moreover, by noting that $u \in \calA_\Gamma$ and taking inverses on both sides of equation~\eqref{eq:quasifactor}, we obtain from Theorem~\ref{thm:quasiauto} that there is $\widetilde{u} \in \calA$ such that
\begin{align*}
    \alpha &=\tau_{0,1}^{\Phi}=\Ad(\tilde u)\circ 
\lmk\widetilde \beta_{\Gamma_2'\setminus\Gamma_1'}^{-1}\circ \tau_{0,1}^{\Phi^{(0)}}
\rmk\\
         &=\Ad(\tilde u)\circ\lmk\widetilde \beta_{\Gamma_2'\setminus\Gamma_1'}^{-1}\circ\alpha_{\Gamma_2\setminus \Gamma_1}\rmk\circ \lmk
\alpha_{\Gamma_1}\otimes\alpha_{\Gamma_2^c}
\rmk \\
         &= \Ad(\widetilde{u}) \circ \widetilde{\Xi} \circ (\alpha_{\Lambda} \otimes \alpha_{\Lambda^c}).
\end{align*}
Here, $\Xi:=\widetilde \beta_{\Gamma_2'\setminus\Gamma_1'}^{-1}\circ\alpha_{\Gamma_2\setminus \Gamma_1}$
is an automorphism on $\caA_{\Gamma_2'}=\caA_{\Lambda_0}$,
$\alpha_\Lambda:=\alpha_{\Gamma_1}$ is an automorphism on
$\caA_{\Lambda}=\caA_{\Gamma_1}$,
and $\alpha_{\Lambda^c}:= \alpha_{\Gamma_2^c} \otimes \operatorname{id}_{\Lambda \setminus \Gamma_2}$
is an automorphism on 
$\caA_{\Lambda^c}$.
\end{proof}

\bibliographystyle{alpha}
\bibliography{References.bib}

\end{document}